\documentclass[a4paper, 12pt]{amsart}
\usepackage{amsmath}
\usepackage{graphicx}
\usepackage{amsfonts}
\usepackage{amssymb}

\newtheorem{theorem}{Theorem}[section]
\newtheorem{proposition}[theorem]{Proposition}
\newtheorem{lemma}[theorem]{Lemma}
\newtheorem{corollary}[theorem]{Corollary}
\theoremstyle{definition}
\newtheorem{definition}[theorem]{Definition}
\newtheorem{remark}[theorem]{Remark}
\newtheorem{example}[theorem]{Example}

\makeatletter
\@namedef{subjclassname@2020}{	\textup{2020} Mathematics Subject Classification}
\makeatother

\begin{document}
\date{2024-7-10}
\title[Extending quantum detailed balance]{Extending quantum detailed balance through optimal transport}
\author{Rocco Duvenhage, Samuel Skosana and Machiel Snyman}
\address{Rocco Duvenhage\\
	Department of Physics\\
	University of Pretoria\\
	Pretoria\\
	South Africa}
\email{rocco.duvenhage@up.ac.za}
\address{Samuel Skosana\\
	Department of Physics\\
	University of Pretoria\\
	Pretoria\\
	South Africa\\
	and DSI-NRF Centre of Excellence
	in Mathematical and Statistical Sciences (CoE-MaSS), South Africa}
\address{Machiel Snyman\\
	Department of Physics\\
	University of Pretoria\\
	Pretoria\\
	South Africa\\
	and	DSI-NRF Centre of Excellence
	in Mathematical and Statistical Sciences (CoE-MaSS), South Africa\\
	(Current address: Fakulteit Natuurwetenskappe, Akademia, Pretoria, South Africa)}

\begin{abstract}
We develop a general approach to setting up and studying classes of quantum
dynamical systems close to and structurally similar to systems having
specified properties, in particular detailed balance. This is done in terms of
transport plans and Wasserstein distances between systems on possibly
different observable algebras.


\end{abstract}

\maketitle
\tableofcontents

\section{Introduction}

The goal of this paper is to develop a natural framework in which one can
study extensions (including weaker forms) of quantum detailed balance
conditions, with the aim of ultimately providing classes of systems with
non-equilibrium steady states, having structural similarities to systems
satisfying conventional detailed balance conditions. Put differently, we
establish a new approach to deviation from detailed balance, even in the
classical case. Our approach is via optimal transport.

One of the paper's main objectives is to develop a point of view that shows
why the bimodule approach to noncommutative optimal transport, as developed in
\cite{D1, D2}, fits perfectly with dynamical systems possessing steady states,
and thus provides a framework for their analysis. At the same time this theory
of optimal transport is developed further, as motivated by our goal above.

As has often been pointed out (see for example \cite{AFQ, AFQ2, AI, BOR, FU}),
to study non-equilibrium successfully, it is of great value to restrict to
classes of non-equilibrium states with sufficient structure to be of interest,
but still amenable to analysis. Our approach promises to provide classes of
systems structurally close to systems satisfying detailed balance, which
should ensure that they have enough structure to build a theory around them,
and still be relevant to physical situations.

The essence of our framework is to extend detailed balance conditions by
considering balance between two systems, in the sense of \cite{DS2}, making
use of couplings of the states of the two systems. However, couplings can also
be interpreted in terms of transport. Optimal transport, through Wasserstein
distances between systems, can then give a quantitative version of the theory
in \cite{DS2}, quantifying how far the two systems are from one another. In
particular, if one of the two systems is a simple or well understood system
satisfying some detailed balance condition, this setup can be used to define
extensions of this condition in the other system. Essentially, some properties
of one system are carried over to the other in a very structured and
quantifiable way. This can be done even if the two systems have different
observable algebras, although one would prefer some natural physical relation
between certain observables of the two systems. Such a relation is in fact
built into the Wasserstein distance, with the mentioned observables acting as
coordinates in the respective systems.

We use a transport plan approach to quadratic Wasserstein distances between
systems (extending Wasserstein distances between states) developed in
\cite{D1, D2}, though we need to extend it to systems on different observable
algebras. It should be noted that while these distances may be actual metrics,
they are more generally only pseudometrics and can even lack symmetry. The
term ``distance'' is therefore used in this generalized sense. The triangle
inequality will always hold, though. This is essential for obtaining certain
bounds on deviation from detailed balance in Subsection \ref{OndAfdAfwVanFB}.

There are other approaches to quantum Wasserstein distances, at least between
states (rather than systems) on the same algebra, some key papers being
\cite{BV, CM1, CGT, CEFZ, DMTL, H1, W}, along with many others. Many of these
use a dynamical (see \cite{BB}) rather than transport plan approach.

However, the approach in \cite{D1, D2} appears most appropriate to fit with
\cite{DS2}, due to its bimodule setup and level of generality. Indeed, the
approach in \cite{D1, D2} was developed in large part with this goal in mind.
In particular, it will be seen that the asymmetry inherent to the bimodule
approach to transport plans, is natural from the perspective of refining or
extending concepts regarding systems.

Optimal transport and Wasserstein distance between classical dynamical systems
do not appear to have been studied anywhere near as extensively as between
measures, but there have been some papers on the topic, the first of these
appearing to be \cite{O}. Even aside from the quantum setup of this paper,
though, the goals and approach in those papers are quite different from ours.

In the literature one can find other approaches to extend quantum detailed
balance to non-equilibrium situations, or to study the deviation from detailed
balance, including via entropy production. See for example \cite{AFQ, AFQ2,
AI, BQ, BQ2, BOR, FR10, FR, FU, TDiehl14}. Our work is an alternative, and we
believe complementary, approach. This paper will not attempt to connect our
approach to those mentioned above or to entropy production. We leave this for
future work. However, we note that by following very different routes than
ours, the papers \cite{Au, DecSak, VS} did connect entropy production to
optimal transport and Wasserstein distance.

Although our main motivation is to extend quantum detailed balance properties,
it appears that our approach is new in the classical case as well, and may
have application there. Moreover, our framework is not restricted only to
extending detailed balance. Although the latter appears to be a very natural
application of the balance and optimal transport approach, it should apply to
other properties of systems with invariant states as well. For example ergodic properties.

Our approach will be formulated in terms of von Neumann algebras, with unital
completely positive (u.c.p.) maps serving as the dynamics of open systems,
although for reasons to be explained, merely unital positive (u.p.) maps will
also be allowed. This places the results of Tomita-Takesaki theory at our
disposal, along with the related theory of bimodules (or correspondences in
the sense of Connes), which are indispensable for the technical underpinnings
of the general theory. In addition, this von Neumann algebra setting is
conceptually very clear, for example in terms of comparison to the classical
case. This combination of conceptual clarity along with powerful technical
tools, is the reason we use the theory of von Neumann algebras as the basic setting.

We nevertheless strive to make the essential ideas accessible to those not
well versed in von Neumann algebras, by also expressing or illustrating many
of the key points in special cases, using examples in terms of $n\times n$
matrix algebras or classical Markov chains on finite sets.

\subsection{A reader's guide}

The reader can get an elementary overview of our ideas, goals and approach,
along with some motivation and intuition, from Section \ref{AfdKlasVb}, while
Examples \ref{VbEindDimStandVorm} and \ref{EindDimModGroep} together with most
of Subsection \ref{OndAfdW} form an entry point to the more general framework,
but for the simple case of $n\times n$ matrix algebras. The bulk of the paper
can then be read as if written for the latter special case, though much of the
technical proofs in Sections \ref{AfdW} and \ref{AfdWEiensk} will not be clear
from this point of view.

Section \ref{AfdBasKon} reviews the basic elements needed in our mathematical
framework in terms of von Neumann algebras, but we also express the key points
in finite dimensions, requiring only some linear algebra. In Section
\ref{AfdKon} we give a conceptual exposition of the approach introduced in
this paper, including an outline of Wasserstein distance. Subsequently
Sections \ref{AfdKlasVb} and \ref{AfdKwantVb} treat simple examples to clarify
the approach and to illustrate how extended or refined detailed balance
conditions are obtained. Section \ref{AfdKlasVb} has been written in such a
way that it should for the most part be accessible without reading the
previous sections. Prior to Section \ref{AfdKwantVb}, a specific form of
quantum detailed balance, namely standard quantum detailed balance with
respect to a reversing operation, is reviewed in Section \ref{AfdFB&KMS},
along with brief general remarks on quantum detailed balance. In order to do
this, the notion of KMS duals needs to be discussed as well. The latter also
plays a role in the general theory of Wasserstein distances. While Section
\ref{AfdFB&KMS} logically belongs before Section \ref{AfdKwantVb}, the latter
is accessible without having read Section \ref{AfdFB&KMS}.

The sections mentioned in the previous paragraph essentially set out, motivate
and illustrate our approach to analyzing systems via (optimal) transport to
and from other (typically simpler) systems. What is still needed is the
theoretical underpinning for a basic tool of the theory, namely Wasserstein
distance, which the next two sections are devoted to. In Section \ref{AfdW} we
give the general definition and derive the basic metric properties of
Wasserstein distances between systems on possibly different von Neumann algebras.

Further properties of Wasserstein distances are then investigated in Section
\ref{AfdWEiensk}, along with their consequences. Symmetries of Wasserstein
distances obtained in Subsection \ref{OndAfdAfwVanFB}, along with the triangle
inequality, show how the Wasserstein distance of a system from a possibly
simpler (say classical) system satisfying detailed balance, can bound how far
the system is from satisfying detailed balance. This is one of the central
conclusions of the paper, and is summarized in Corollary \ref{fbAfskat}. Next,
keep in mind that a Wasserstein distance can be an asymmetric pseudometric. In
particular, zero distance between two systems does not mean that they are the
same. Subsection \ref{OndAfdW&GemStruk} investigates how zero Wasserstein
distance between two systems relates to common structure in the two systems.

The paper concludes with a short synopsis together with a brief discussion of
further work and possibilities in Section \ref{AfdEinde}.

\section{Basic concepts\label{AfdBasKon}}

Here we provide and review relevant concepts and technical definitions used in
subsequent sections. Some of these definitions, in particular that of systems,
only appear in a special form in this section, with more general versions
being treated in Section \ref{AfdW}. This is done for clarity and to emphasize
the core ideas. For example, here we focus on the case of discrete time.

\subsection{Preliminary definitions, conventions and
notations\label{OndAfdDef}}

Throughout the paper, $A$, $B$ and $C$ denote $\sigma$-finite von Neumann
algebras in standard form, except when stated otherwise. See \cite{Ar74,
Con74, Ha75} for the original papers on standard forms, and \cite[Section
2.5.4]{BR1} for a very good treatment. We denote identity operators on Hilbert
spaces by $1$. In particular, the unit $1_{A}$ of a von Neumann algebra $A$,
represented on a Hilbert space $H_{A}$, will usually be denoted by $1$. Denote
the set of faithful normal states on $A$ by $\mathfrak{F}(A)$. The normal
states are exactly those given by density matrices (see \cite[Theorem
2.4.21]{BR1}).

By the universality of a certain natural positive cone $\mathcal{P}_{A}\subset
H_{A}$ associated to $A$, the modular conjugation associated to all $\mu
\in\mathfrak{F}(A)$ is the same and will be denoted $J_{A}$. It is an
anti-unitary operator in $H_{A}$ with $J_{A}^{2}=1$, and we can define the
linear map%
\[
j_{A}:=J_{A}(\cdot)^{\ast}J_{A}:\mathcal{B}(H_{A})\rightarrow\mathcal{B}%
(H_{A})
\]
on the bounded linear operators $\mathcal{B}(H_{A})$, i.e., $j_{A}%
(a):=J_{A}a^{\ast}J_{A}$ for every bounded linear $a:H_{A}\rightarrow H_{A}$.
Then $j_{A}$ is easily seen to be an anti-$\ast$-automorphism. I.e., $j_{A}$
is linear, bijective and $j_{A}(a^{\ast})=j_{A}(a)^{\ast}$, but $j_{A}(ab)=$
$j_{A}(b)j_{A}(a)$ for all $a,b\in\mathcal{B}(H_{A})$. From this we can
define
\[
\mu^{\prime}=\mu\circ j_{A}\in\mathfrak{F}(A^{\prime})
\]
on the commutant $A^{\prime}$ of $A$ for every $\mu\in\mathfrak{F}(A)$, i.e.,
$\mu^{\prime}(a^{\prime})=\mu(j_{A}(a^{\prime}))$ for all $a^{\prime}\in
A^{\prime}$. Since every $\mu\in\mathfrak{F}(A)$ is given by a (unique) cyclic
and separating vector $\Lambda_{\mu}\in\mathcal{P}_{A}$ for $A$, through
\[
\mu(a)=\left\langle \Lambda_{\mu},a\Lambda_{\mu}\right\rangle
\]
for all $a\in A$ in terms of $H_{A}$'s inner product, we can represent
$\mu^{\prime}$ by
\[
\mu^{\prime}(a^{\prime})=\left\langle \Lambda_{\mu},a^{\prime}\Lambda_{\mu
}\right\rangle
\]
for all $a^{\prime}\in A^{\prime}$.

\begin{example}
\label{VbEindDimStandVorm}
This example briefly recounts a standard form for the von Neumann algebra
$M_{n}$ of $n\times n$ complex matrices. First note that any faithful state
(necessarily normal in finite dimensions) on $M_{n}$ is given by
$\mu(a)=\operatorname*{Tr}(\rho_{\mu}a)$ for all $a\in M_{n}$, where
$\rho_{\mu}\in M_{n}$ is some density matrix whose eigenvalues are strictly
positive. However, in the general theory we don't use this representation
directly. We rather use a standard form, one formulation of which is as
follows. Write%
\[
H_{A}=\mathbb{C}^{n}\otimes_{s}\mathbb{C}^{n}=M_{n},
\]
where elementary tensors in $H_{A}$ are written in the form $x\otimes
_{s}y:=xy^{\intercal}$ in terms of usual matrix multiplication, for column
vectors $x,y\in\mathbb{C}^{n}$ and with the row vector $y^{\intercal}$ being
the transpose of $y$. The inner product of $H_{A}$ is%
\[
\left\langle X,Y\right\rangle :=\operatorname*{Tr}(X^{\ast}Y)
\]
where $X^{\ast}=X^{\dagger}$ is the usual adjoint of the matrix $X\in M_{n}$.
We represent $M_{n}$ on $H_{A}$ as
\[
A=M_{n}\otimes_{s}1_{n}%
\]
with $1_{n}$ the $n\times n$ identity matrix. Here we use the notation
$\otimes_{s}$ to represent the tensor product $M_{n}\otimes M_{n}$ on
$H_{A}=M_{n}$ via%
\[
(a\otimes_{s}b)X:=aXb^{\intercal}%
\]
for all $a,b\in M_{n}$. An alternative but equivalent representation is the
``usual'' tensor product $H_{A}=\mathbb{C}^{n}\otimes\mathbb{C}^{n}$ taken as
the Kronecker product, but $\otimes_{s}$ emphasizes the bimodule structure
inherent to our setting and it is more convenient for certain purposes, as
will be seen in Example \ref{VbWvirMn}. (Also see \cite[Section 2]{DS1}.)

In this representation we use the notation $\pi(a):=a\otimes_{s}1_{n}$ and
$\pi^{\prime}(b)=1_{n}\otimes_{s}b$ for $a,b\in M_{n}$, in terms of which we
have%
\[
\mu(\pi(a))=\operatorname*{Tr}(\rho_{\mu}a)=\left\langle \Lambda_{\mu}%
,\pi(a)\Lambda_{\mu}\right\rangle
\]
for all $a\in M_{n}$, where
\[
\Lambda_{\mu}:=\rho_{\mu}^{1/2}\in H_{A}.
\]
Furthermore, $A$'s commutant is
\[
A^{\prime}=1_{n}\otimes_{s}M_{n},
\]
while $J_{A}$ is given by%
\[
J_{A}\pi(a)\Lambda_{\mu}=\pi^{\prime}((a^{\ast})^{\intercal})\Lambda_{\mu}%
\]
(which is indeed independent of the faithful state $\mu$) and one has
\[
j_{A}(\pi(a))=\pi^{\prime}(a^{\intercal})\text{ \ \ and \ \ }j_{A}(\pi
^{\prime}(a))=\pi(a^{\intercal})
\]
for $a\in M_{n}$. Note that this means that
\[
\mu^{\prime}(\pi^{\prime}(a))=\mu(\pi(a^{\intercal}))=\operatorname*{Tr}%
(\rho_{\mu}a^{\intercal})=\operatorname*{Tr}(\rho_{\mu}^{\intercal}a)
\]
for every $a\in M_{n}$. It is easily shown that $\pi(a^{\ast})=\pi(a)^{\ast}$
and $\pi^{\prime}(a^{\ast})=\pi^{\prime}(a)^{\ast}$, i.e., $\left\langle
X,\pi(a)Y\right\rangle =\left\langle \pi(a^{\ast})X,Y\right\rangle $, and
similarly for $\pi^{\prime}$. The natural positive cone is defined (see
\cite[Definition 2.5.25]{BR1}) to be%
\[
\mathcal{P}_{A}=\{\pi(a)j_{A}(\pi(a)^{\ast})1_{n}:a\in M_{n}\}=\{aa^{\ast
}:a\in M_{n}\},
\]
the trace one elements of which give exactly all the density matrices in
$M_{n}$, with $1_{n}$ essentially serving as a ``reference'' element of
$H_{A}$.

Keeping these points in mind, one can fairly easily specialize the general von
Neumann algebraic setting above to $M_{n}$.

To conclude this example, note that when using the usual tensor product
$\otimes$ instead of $\otimes_{s}$, all the above still goes through, as they
are simply different representations of the tensor product, but the resulting
representation of $\Lambda_{\mu}$ becomes less convenient for our purposes.
\end{example}

\begin{remark}
Strictly speaking, everything we have done so far continues to work even if
the state $\mu$ is not faithful (as $\mathcal{P}_{A}$ includes all normal
states). However, certain duals (see Definition \ref{DefDuaal}) are going to
play a central role in our work. Their definition requires cyclic vectors for
$A$ in $\mathcal{P}_{A}$, and these are necessarily also separating for $A$,
making the corresponding states on $A$ faithful. Hence our focus on faithful states.
\end{remark}

Dynamical maps in this paper will be unital positive linear maps (u.p. maps,
for short) between von Neumann algebras. In quantum physics one is generally
interested in unital completely positive linear maps (u.c.p. maps), but our
mathematical framework (and indeed that of \cite{D2}) works for u.p. maps. In
addition, for certain kinds of detailed balance, we want to be able to view
reversing operations, which are not completely positive in the noncommutative
case, as a form of dynamics as well (see Sections \ref{AfdFB&KMS} and
\ref{AfdW}). We therefore need to cover maps which are merely positive, in our
framework. A key notion for us will be that of a dual of a dynamical map. A
form general enough for this paper, is as follows.

\begin{definition}
\label{DefDuaal}Given a u.p. map $E:A\rightarrow B$ satisfying $\nu\circ
E=\mu$ for some $\mu\in\mathfrak{F}(A)$ and $\nu\in\mathfrak{F}(B)$, its
\emph{dual} (w.r.t. $\mu$ and $\nu$)
\[
E^{\prime}:B^{\prime}\rightarrow A^{\prime}%
\]
is defined by
\[
\left\langle \Lambda_{\mu},aE^{\prime}(b^{\prime})\Lambda_{\mu}\right\rangle
=\left\langle \Lambda_{\nu},E(a)b^{\prime}\Lambda_{\nu}\right\rangle
\]
for all $a\in A$ and $b^{\prime}=B^{\prime}$.
\end{definition}

Note that according to \cite[Proposition 3.1]{AC}, $E^{\prime}$ exists, and it
is a u.p. map satisfying $\mu^{\prime}\circ E^{\prime}=\nu^{\prime}$. Under
the given assumptions, the maps $E$ and $E^{\prime}$ are necessarily normal,
i.e., $\sigma$-weakly continuous. If $E$ is u.c.p., then so is $E^{\prime}$.
The reader can refer to \cite{AC} for the theory behind such duals, also
summarized in \cite[Section 2]{DS2}. We do not include $\mu$ and $\nu$ in the
notation for the dual, as it will always be clear form context, and would make
the notation unnecessarily cumbersome.

We also point out that unitality and invariance can be viewed as dual
properties, in the following sense: Assume that two linear maps
$E:A\rightarrow B$ and $F:B^{\prime}\rightarrow A^{\prime}$ satisfy
\[
\left\langle \Lambda_{\mu},aF(b^{\prime})\Lambda_{\mu}\right\rangle
=\left\langle \Lambda_{\nu},E(a)b^{\prime}\Lambda_{\nu}\right\rangle
\]
for all $a\in A$ and $b^{\prime}\in B^{\prime}$. If $\nu\circ E=\mu$, then
setting $b^{\prime}=1$, we see that $F(1)=1$, since $\Lambda_{\mu}$ is cyclic
and separating. The converse is trivial. Similarly for $\mu\circ F=\nu$ versus
$E(1)=1$, by the fact that $\Lambda_{\nu}$ is cyclic and separating for $B$ in
terms of its Hilbert space.

The main mathematical object of concern in this paper is a system. A
\emph{system} $\mathbf{A}$ is defined as
\[
\mathbf{A}=(A,\alpha,\mu)
\]
with $A$ a von Neumann algebra (as above), $\alpha:A\rightarrow A$ a u.p. map,
and $\mu\in\mathfrak{F}(A)$ such that $\mu\circ\alpha=\mu$.

The \emph{dual} of a system $\mathbf{A}$ is defined as the system
\[
\mathbf{A}^{\prime}=(A^{\prime},\alpha^{\prime},\mu^{\prime}).
\]
These notations will be used as a conventions in the rest of the paper, along
with
\[
\mathbf{B}=(B,\beta,\nu)
\]
for a system on $B$. In Section \ref{AfdW}, however, the definition of a
system will be generalized.

\subsection{Transport plans between systems\label{OndAfdOP}}

The main method of studying systems in this paper is through transport between
systems. Given two systems $\mathbf{A}$ and $\mathbf{B}$, a \emph{transport
plan }$\omega$\emph{ from }$\mathbf{A}$\emph{ to }$\mathbf{B}$ is a state
$\omega$ on the algebraic tensor product $A\odot B^{\prime}$ such that%
\begin{equation}
\omega(a\otimes1)=\mu(a)\text{ \ and \ }\omega(1\otimes b^{\prime}%
)=\nu^{\prime}(b^{\prime}) \label{kop}%
\end{equation}
and
\begin{equation}
\omega(\alpha(a)\otimes b^{\prime})=\omega(a\otimes\beta^{\prime}(b^{\prime}))
\label{oordPl}%
\end{equation}
for all $a\in A$ and $b^{\prime}\in B^{\prime}$. We express this by the
notation%
\[
\mathbf{A}\omega\mathbf{B}%
\]
and denote the set of all such transport plans $\omega$ by
\[
T(\mathbf{A},\mathbf{B}).
\]

The notion $\mathbf{A}\omega\mathbf{B}$ was introduced in \cite{DS2} under the
name ``balance'' as an extension of certain detailed balance conditions. The
latter would be formulated in terms of $B=A$, with $\beta$ being some
variation on the dual $\alpha^{\prime}$ of $\alpha$, while $\omega$ is a
specific state; this will be reviewed in Section \ref{AfdFB&KMS}. At the same
time it extends the notion of a joining in ergodic theory, which allows one to
study certain qualitative and structural aspects of dynamical systems,
essentially by comparing different systems (see \cite{G} for the classical
theory, and \cite{D(B1), D(B2), D(B3), BCM, BCM2, BCM3} for a noncommutative version).
Both these points are relevant in this paper as well. However, we take the
(optimal) transport point of view in this paper, hence we refer to
$\mathbf{A}\omega\mathbf{B}$ as transport from $\mathbf{A}$ to $\mathbf{B}$,
rather than as balance. A special case of this was used in \cite{D1},
specifically for modular dynamics (from Tomita-Takesaki theory), to obtain
symmetry of Wasserstein metrics on states.

If the state $\omega$ is only required to satisfy the coupling property
(\ref{kop}), then it is called a \emph{transport plan} $\omega$ \emph{from} 
$\mu$ \emph{to} $\nu$. The set of all of these is denoted by
\[
T(\mu,\nu).
\]
This is in fact a fairly direct extension of classical transport plans, a nice
discussion of which, including why it is thought of as transport, can be found
in the introduction to \cite{V1}. There a transport plan $\pi$ from one
probability measure $\mu$ (on a space $X$) to another $\nu$ (on a space $Y$)
is simply a coupling of the two measures, i.e., $\pi(U\times Y)=\mu(U)$ and
$\pi(X\times V)=\nu(V)$, with the following interpretation: $\pi(U\times V)$
is the ``amount'' of probability transported from $U\subset X$ to $V\subset Y$.

The reason for involving the dual of a system, in particular the commutant of
the von Neumann algebra, is to alleviate obstacles due to noncommutativity.
This is what ultimately allows one to obtain the usual metric properties for
Wasserstein distances, specifically the triangle inequality and zero distance
between a system and itself. It is at the core of the bimodule approach to
Wasserstein distances, and serves as the key to the mathematical foundation of
this paper.

It is well known in quantum physics, in particular in the context of quantum
information, that a transport plan $\omega$ corresponds to a u.c.p. map i.e.,
a channel, $E_{\omega}:A\rightarrow B$, although the terminology ``transport
plan'' is not standard in this context, nor the interpretation as transport
(as an analogy to the classical case). See in particular \cite{Choi} for the
finite dimensional case, as is typically used in quantum information. A
general form of this correspondence was discussed and used in \cite{BCM, DS2},
making use of \cite{AC}. In the context of transport plans, this of course
means that a transport plan can be represented as a u.c.p. map, as has been
pointed out by \cite{dePT}. In this way a transport plan can indeed be viewed
as a dynamical map from one system to another, in the conventional quantum
mechanical sense. The relevant technical details will be presented in Section
\ref{AfdW}. In the classical setting the correspondence between couplings of
probability measures and Markov operators was studied in a special case (for
couplings of a measure with itself) in \cite{B}, and more generally in
\cite{MT, BM}, though it does not appear to be widely used in classical
optimal transport.

\section{Conceptual exposition\label{AfdKon}}

In this section we describe the basic conceptual approach of this paper,
including a preliminary discussion of Wasserstein distance without proofs.
This also provides the basics on Wasserstein distance needed in the examples
studied in the subsequent sections. The general theory of Wasserstein distance
will be taken up in Section \ref{AfdW}. The key ideas of this section are
illustrated in a simple classical case in Section \ref{AfdKlasVb}, which the
reader can read before this section with minimal loss in continuity.

\subsection{Extending dynamical properties by optimal
transport\label{OndAfdUitbrDmvOV}}

Given a property which a system may have or may not have, we aim to define an
extended or generalized version of this property for a system $\mathbf{A}$, by
requiring $\mathbf{A}$ to satisfy%
\[
\mathbf{A}\omega\mathbf{B}%
\]
for some system $\mathbf{B}$ which does satisfy the original property. In
particular, our goal is to define extensions of detailed balance conditions in
this way. Keep in mind that when $\omega$ is a product state, $\mathbf{A}%
\omega\mathbf{B}$ is trivial. To ensure that the condition $\mathbf{A}%
\omega\mathbf{B}$ is nontrivial and indeed as strong as possible, we are going
to frame this approach as an optimal transport problem. Wasserstein distances
will deliver optimal transport plans, which in our approach are viewed as the
strongest cases of the condition $\mathbf{A}\omega\mathbf{B}$. This
optimization and Wasserstein distances will be discussed in more detail in the
next subsection.

Of course, one may as well instead consider%
\[
\mathbf{B}\psi\mathbf{A,}%
\]
which gives an alternative way to extend properties from $\mathbf{B}$ to
$\mathbf{A}$, by viewing it as an optimal transport problem from $\mathbf{B}$
to $\mathbf{A}$. The two points of view are complementary, but not equivalent.
It will indeed be seen later on that the Wasserstein distances involved are
not necessarily symmetric.

We are going to motivate this transport plan interpretation by viewing it as a
refinement of the transfer of probabilities in $\mathbf{A}$ itself during each
step in its discrete time evolution.

First notice that the two basic properties satisfied by the completely
positive map $\alpha$, namely invariance $\mu\circ\alpha$ and unitality, which
are dual to one another, can respectively be expressed as $\mathbf{A}%
\omega\mathbf{B}$ and $\mathbf{B}\psi\mathbf{A}$, if $\mathbf{B}$ is taken as
a 1-point system. By the latter terminology we mean that $B=\mathbb{C}$. In
terms of the algebraic tensor product $A\odot B^{\prime}$, one then
necessarily has%
\[
\omega=\mu\odot\nu^{\prime}\text{ \ and \ }\psi=\nu\odot\mu^{\prime},
\]
which means that $\mathbf{A}\omega\mathbf{B}$ merely says
\[
\mu\circ\alpha(a)b'
= \omega \left(  \alpha(a) \otimes b' \right)
= \omega \left(  a \otimes \beta'(b') \right)  
= \mu(a)b'
\]
for all $a\in A$ and $b^{\prime}\in B^{\prime}=\mathbb{C}$, i.e., $\mu
\circ\alpha=\mu$, while $\mathbf{B}\psi\mathbf{A}$ similarly expresses
$\mu^{\prime}\circ\alpha^{\prime}=\mu^{\prime}$, which is dual (and
equivalent) to $\alpha(1)=1$.

In terms of transport between $\mathbf{A}$ and $\mathbf{B}$, this means that
transport from $\mathbf{A}$ to a single point expresses state invariance,
while transport from a single point to $\mathbf{A}$ expresses unitality of
$\alpha$. In the classical case where $\mathbf{A}$ consists of a finite number
of points, one can in fact identify the 1-point system as an arbitrary point
in $\mathbf{A}$, as will be seen in the next section.

By replacing the 1-point system by a general $\mathbf{B}$, we in effect lift
transfer of probabilities to or from a single point, to a more general form of
transport to or from $\mathbf{B}$. The condition $\mathbf{A}\omega\mathbf{B}$
then refines state invariance, while $\mathbf{B}\psi\mathbf{A}$ refines
unitality. This refinement of these two dual properties each places further
restrictions on $\mathbf{A}$, determined by the properties of $\mathbf{B}$.

By enforcing optimal transport through a Wasserstein distance, these
refinements are made as strong as possible, quantifying how strongly the
properties of $\mathbf{B}$ are reflected in $\mathbf{A}$. This can then be
used to explore generalized versions of properties of $\mathbf{B}$ in an
arbitrary system $\mathbf{A}$.

When $\mathbf{A}$ and $\mathbf{B}$ are systems on the same observable algebra,
and the Wasserstein distance is a (typically asymmetric) metric, then the
extreme case of zero distance means that $\mathbf{A}=\mathbf{B}$, i.e.,
$\mathbf{A}$ has all of $\mathbf{B}$'s properties. Other cases are typically
weaker versions of this extreme case, which can for example provide weakened
forms of detailed balance in $\mathbf{A}$.


Placing restrictions on a system $\mathbf{A}$ through (optimal) transport to
or from a ``model system'' $\mathbf{B}$ is natural, as such conditions generalize
probability transfer between individual points in a classical system. When
$\mathbf{B}$ is assumed to satisfy some form of detailed balance, this
provides a way to formulate extensions of this detailed balance condition in
$\mathbf{A}$. As this is done in a very structured way, it has the potential
to provide classes of systems having sufficient structure to include natural,
yet mathematically accessible, non-equilibrium systems.

Note that as $\mathbf{A}\omega\mathbf{B}$ and $\mathbf{B}\psi\mathbf{A}$
relate to two different properties of $\mathbf{A}$, namely $\mu\alpha=\mu$ and
$\alpha(1)=1$ respectively, we would not expect perfect symmetry between the
two situations, even classically. This is indeed what is found. In particular,
our Wasserstein distances typically turn out to be asymmetric, since the two
directions (the distance from $\mathbf{A}$ to $\mathbf{B}$ versus from
$\mathbf{B}$ to $\mathbf{A}$) relate to different (though dual) aspects of dynamics.

For simplicity, in the examples of this paper the case where $\mathbf{B}$ is a
classical system on a finite set of points will be emphasized. On the one hand
this simplifies calculations involved, but it also simplifies and clarifies
how to identify detailed balance type properties consequently appearing in
$\mathbf{A}$.

\subsection{Wasserstein distances\label{OndAfdW}}

We intend to define a Wasserstein distance from $\mathbf{A}$ to $\mathbf{B}$
in analogy to quadratic Wasserstein distance between probability measures on
$\mathbb{R}^{d}$, while also taking the dynamics into account. Due to the
dynamics and noncommutativity, such distances are typically not symmetric,
i.e., the distance from $\mathbf{B}$ to $\mathbf{A}$ may differ than that from
$\mathbf{A}$ to $\mathbf{B}$. In Section \ref{AfdW} it will be seen that one
can define refined versions of these Wasserstein distances, including
symmetric versions. In this section, however, only the most basic form will be
outlined in a special case. The general development of the theory is postponed
until Section \ref{AfdW}.

We assume that $\mathbf{A}$ and $\mathbf{B}$ are equipped with ``$d$%
-dimensional coordinate systems'' expressed as
\[
k=(k_{1},...,k_{d})\text{ \ \ and \ \ }l=(l_{1},...,l_{d})
\]
respectively, where $k_{i}\in A$ and $l_{i}\in B$ for $i=1,...,d$. Here we do
not place any restrictions on the choice of $k$ and $l$, for example,
$l_{1}=l_{2}$ is allowed.

When we confine ourselves to the case of finite dimensional or abelian von
Neumann algebras, then these coordinate systems define a cost $I(\omega)$ for
any transport plan $\omega\in T(\mathbf{A},\mathbf{B})$ from $\mathbf{A}$ to
$\mathbf{B}$ by
\begin{equation}
I(\omega)=\omega\left(  \sum_{i=1}^{d}\left|  k_{i}\otimes1-1\otimes(S_{\nu
}l_{i}^{\ast}S_{\nu})\right|  ^{2}\right)  \label{eindDimKoste}%
\end{equation}
in analogy to the classical transport cost for a transport plan $\pi$ between
two measures on the same subset $X$ of $\mathbb{R}^{d}$
\begin{equation}
I(\pi)=\int_{X\times X}\left\|  x-y\right\|  ^{2}d\pi(x,y) \label{klasKoste}%
\end{equation}
in terms of the Euclidean norm on $\mathbb{R}^{d}$. Wasserstein distance is
obtained as the minimum of the square root of the cost.

Note, however, the inclusion of the conjugate linear operator $S_{\nu}%
:H_{B}\rightarrow H_{B}$ from Tomita-Takesaki theory, defined by%
\[
S_{\nu}b\Lambda_{\nu}:=b^{\ast}\Lambda_{\nu}.
\]
The inclusion of this operator ensures that we get the usual metric
properties, in particular the triangle inequality and zero distance between a
system and itself. For more general von Neumann algebras, $S_{\nu}$ can be
unbounded, leading to technical complications which we can circumvent in
Section \ref{AfdW} by an alternative approach containing (\ref{eindDimKoste})
as special case. For the examples discussed in Sections \ref{AfdKlasVb} and
\ref{AfdKwantVb}, the formulation (\ref{eindDimKoste}) will be sufficient
though. In addition it gives a bird's-eye view of both the similarities and
the differences between the classical and noncommutative cases.

In relation to the latter point, notice that for abelian $B$, finite
dimensional or not, one has $S_{\nu}b^{\ast}S_{\nu}=b$ for all $b\in B$. In
this way the classical cost (\ref{klasKoste}) can be recovered as a special
case of the form in (\ref{eindDimKoste}), at least for a bounded subset $X$ of
$\mathbb{R}^{d}$.

Of course, as opposed to the case of Wasserstein distance between measures on
the same space $X$, (\ref{eindDimKoste}) is also set up to provide costs and
consequently distances between systems on different von Neumann algebras. This
is of central importance to our aims, allowing one to analyze a system by
comparing it to simpler or well understood systems, including those on other algebras.

The Wasserstein distance from $\mathbf{A}$ to $\mathbf{B}$ corresponding to
the chosen coordinate systems, is defined as
\begin{equation}
W(\mathbf{A},\mathbf{B}):=\inf_{\omega\in T(\mathbf{A},\mathbf{B})}%
I(\omega)^{1/2}. \label{spesW}%
\end{equation}
This infimum is reached, as will be seen in Section \ref{AfdW}, therefore
optimal transport plans (i.e., plans with minimum cost) always exist. The key
point in this definition, though, is that the infimum is over the set of all
transport plans from $\mathbf{A}$ to $\mathbf{B}$, making this Wasserstein
distance a natural way to to obtain optimal forms of the transport conditions
in the previous subsection.

To apply this in practice and particular examples, one should keep in mind
that the $k$ and $l$ appearing in a Wasserstein distance are viewed as
``coordinates'' measuring particular aspects of $\mathbf{A}$ and $\mathbf{B}$
respectively. In particular, they may be chosen to be relevant observables in
$\mathbf{A}$ and $\mathbf{B}$. In addition, for every $p$ one can aim to
choose $k_{p}$ and $l_{p}$ to measure corresponding quantities in $\mathbf{A}$
and $\mathbf{B}$, say spin in some direction of a particular spin site in
each, in order for the difference between $k_{p}$ and $l_{p}$ in
(\ref{eindDimKoste}) to have a sensible physical meaning.


A cardinal point regarding (\ref{eindDimKoste}) is that zero cost between a
quantum system and itself can always be attained by using a maximally
entangled state as the transport plan. This is why we have%
\[
W(\mathbf{A},\mathbf{A})=0
\]
as one would expect of a distance function. This appears to be a simple point,
but it is subtler than one may initially expect. The key aspect of this is
illustrated as part of the next example.

\begin{example}
\label{VbWvirMn}
Consider the standard form of $M_{n}$ as described in Example
\ref{VbEindDimStandVorm}. Note that we can identify $A\odot A^{\prime}$ with
$M_{n}\otimes_{s}M_{n}$ and $\pi(a)\otimes\pi^{\prime}(b)$ with $a\otimes
_{s}b$. To see that our setup indeed leads to $W(\mathbf{A},\mathbf{A})=0$ for
the von Neumann algebra $M_{n}$, consider the entangled state $\omega
=\delta_{\mu}$ of $\mu$ and $\mu^{\prime}$, defined on $A\odot A^{\prime}$ by%
\[
\delta_{\mu}(a\otimes_{s}b):=\left\langle \Lambda_{\mu},(a\otimes_{s}%
b)\Lambda_{\mu}\right\rangle =\operatorname*{Tr}(\rho_{\mu}^{1/2}a\rho_{\mu
}^{1/2}b^{\intercal})
\]
for $a,b\in M_{n}$. This is a pure state, since the vector $\Lambda_{\mu}\in
H_{A}$ is a pure state of the composite system, reducing to $\mu$ and
$\mu^{\prime}$ respectively, and in this sense we view $\delta_{\mu}$ as a
maximally entangled state of $\mu$ and $\mu^{\prime}$. (When $\rho_{\mu}$ is
diagonal with diagonal entries $p_{1},...,p_{n}$, we in fact have the
conventional and well known representation $\sum_{i=1}^{n}\sqrt{p_{i}}%
e_{i}\otimes e_{i}$ of $\Lambda_{\mu}$, but our standard form above is more
convenient below for the general case.)

Note that for $a,b\in M_{n}$ we have
\[
S_{\mu}^{\ast}\pi(a)^{\ast}S_{\mu}\pi(b)\Lambda_{\mu}=\pi^{\prime}((\rho_{\mu
}^{-1/2}a\rho_{\mu}^{1/2})^{\intercal})\pi(b)\Lambda_{\mu},
\]
hence $S_{\mu}^{\ast}\pi(a)^{\ast}S_{\mu}=\pi^{\prime}((\rho_{\mu}^{-1/2}%
a\rho_{\mu}^{1/2})^{\intercal})\in A^{\prime}$, which is the first key point
of this example. Furthermore, $S_{\nu}\pi(a^{\ast})S_{\nu}\Lambda_{\mu}%
=\pi(a)\Lambda_{\mu}$, therefore%
\[
\delta_{\mu}(\left|  \pi(a)-S_{\nu}\pi(a^{\ast})S_{\nu}\right|  ^{2})=0.
\]
Hence $\delta_{\mu}\in T(\mu,\mu)$ is a transport plan having zero cost. This
is what makes $W(\mathbf{A},\mathbf{A})=0$ possible, concluding this example.
\end{example}

Our general theory of Wasserstein distances and their basic metric properties,
including $W(\mathbf{A},\mathbf{A})=0$ and the triangle inequality, will be
treated in Section \ref{AfdW}. In some other approaches to Wasserstein
distances via transport plans (in the special case of states, rather than
systems), the property $W(\mu,\mu)=0$ and the usual triangle inequality are not obtained.
For example \cite{GMP} and \cite{dePT}, even though the latter follows a
similar approach to cost as this paper. It should be noted that these papers
nevertheless manage to make progress on their applications, and indeed
\cite{GPTV} has pointed out that a lack of the usual metric properties is
actually of value in certain contexts. A nice discussion and enlightening
illustrations of the obstacles to metric properties for quantum Wasserstein
distances via transport plans can be found in \cite{AF}.

For our purposes though, it is most natural to have the triangle inequality
(see for example its use in Corollary \ref{fbAfskat}) and $W(\mathbf{A}%
,\mathbf{A})=0$. On the other hand, asymmetry of a Wasserstein distance has a
natural place in our theory, as already hinted in Subsection
\ref{OndAfdUitbrDmvOV}. This will also be seen in the next section. The
bimodule setup for Wasserstein metrics indeed provide symmetric and asymmetric
Wasserstein distances, depending on the restrictions we place on the transport
plans, as has been discussed in \cite{D1, D2} and will be seen in Section
\ref{AfdW}.

\section{A simple classical example\label{AfdKlasVb}}

We now clarify the approach in the previous section further by illustrating it
with classical systems, each on a finite number of points. In this case the
transfer of probabilities are easy to follow, making the setup in Subsection
\ref{OndAfdUitbrDmvOV} particularly transparent.

\subsection{The setting\label{OndAfdKlasVbOpset}}

Represent the von Neumann algebra $A$ of complex-valued functions on $m$
points as $A=\mathbb{C}^{m}$, viewed as column matrices. Then $A$ is an
algebra with involution (and indeed an abelian von Neumann algebra) by using
entry-wise operations; for example the involution is given by taking the
complex conjugate of each entry. We label the points $1,...,m$. A standard
form of $A$ is obtained by representing the entries of $a\in A\,$\ along the
diagonal of an $m\times m$ matrix acting on $H_{A}=\mathbb{C}^{m}$, or simply
letting $A$ act entry-wise on $H_{A}$, so essentially we have $A$ in standard
form (with $A^{\prime}=A$ and $j_{A}:A\rightarrow A$ the identity map, hence
not relevant in this section). Similarly, for $B$ on $n$ points.

We consider Markov chains on these $m$ and $n$ point sets, respectively given
by the $m\times m$ and $n\times n$ transition matrices%
\[
\alpha=\left[  \alpha_{pq}\right]  \text{ \ and \ }\beta=\left[  \beta
_{rs}\right]  ,
\]
where
\begin{equation}
\alpha_{pq}\geq0\text{ \ and \ }\sum_{q=1}^{m}\alpha_{pq}=1,
\label{oorgWaarskDef}%
\end{equation}
and similarly for $\beta$. Note that $\alpha_{pq}$ is the probability for a
transition to occur from point $p$ to point $q$.

The states $\mu$ and $\nu$ on $A$ and $B$ respectively, are represented as row
matrices%
\[
\mu=\left[
\begin{array}
[c]{rcr}%
\mu_{1} & \cdots & \mu_{m}%
\end{array}
\right]  \text{ \ and \ }\nu=\left[
\begin{array}
[c]{ccc}%
\nu_{1} & \cdots & \nu_{n}%
\end{array}
\right]  ,
\]
with $\mu_{p}>0$ and $\mu_{1}+...+\mu_{m}=1$ and correspondingly for $\nu$.
These are assumed to be invariant,
\[
\mu\alpha=\mu\text{ \ and \ }\nu\beta=\nu,
\]
in terms of usual matrix multiplication, i.e., $\mu\circ\alpha=\mu\alpha$.
Correspondingly $\alpha$ acts on $A$ through usual matrix multiplication,
$\alpha(a)=\alpha a$, hence we have a system $\mathbf{A}=(A,\alpha,\mu)$.

The system $\mathbf{B}=\left(  B,\beta,\nu\right)  $ is said to be in
\emph{detailed balance}, when%
\[
\nu_{r}\beta_{rs}=\nu_{s}\beta_{sr},
\]
for all $r,s=1,...,n$, i.e., the transfer of probability from point $r$ to
point $s$, is the same as from $s$ to $r$. This is a strictly stronger
property than invariance, except for $n=2$, in which case the two properties
are equivalent.

In the current representation, the dual (from Subsection \ref{OndAfdDef}) of
$\mathbf{A}$ is the system $\mathbf{A}^{\prime}=(A,\alpha^{\prime},\mu)$ with
$\alpha^{\prime}$ given by%
\[
\alpha_{pq}^{\prime}=\frac{\mu_{q}}{\mu_{p}}\alpha_{qp}%
\]
for all $p,q=1,...,m$. This tells us that $\mathbf{A}'$ is the
time-reversal of $\mathbf{A}$. Detailed balance of $\mathbf{B}$ can thus be
expressed as
$\beta'=\beta$, i.e.,
\[
\mathbf{B}' = \mathbf{B}.
\]

\subsection{Extending probability transfer to or from a point}

The total transfer of probability from the whole of $\mathbf{A}$ to a single
point $q$ in $\mathbf{A}$, can be written as%
\begin{equation}
\sum_{p=1}^{m}\mu_{p}\alpha_{pq}=\mu_{q}, \label{KlasVervoerNaPunt}%
\end{equation}
which simply expresses the invariance $\mu\alpha=\mu$. In Subsection
\ref{OndAfdUitbrDmvOV} we pointed out that $\mathbf{A}\omega\mathbf{B}$
generalizes this condition, being viewed as probability transfer from
$\mathbf{A}$ to $\mathbf{B}$, 
instead of from $\mathbf{A}$ to a single point in
$\mathbf{A}$. We interpreted $\mathbf{A}\omega\mathbf{B}$ as a refinement of
$\mu\alpha=\mu$. This will now be seen concretely for the current classical
systems. (We now write $\otimes$, instead of $\odot$ as in Subsection
\ref{OndAfdUitbrDmvOV}, since in finite dimensions we do not need to
distinguish between algebraic tensor products and completions thereof.)

Using the Kronecker product convention%
\[
\mu\otimes\nu=\left[
\begin{array}
[c]{rrr}%
\mu\nu_{1} & \cdots & \mu\nu_{n}%
\end{array}
\right]  ,
\]
a transport plan $\omega$ from $\mu$ to $\nu$ can be represented as
\[
\omega=\left[
\begin{array}
[c]{rrrcrrr}%
\omega_{11} & \cdots & \omega_{m1} & \cdots & \omega_{1n} & \cdots &
\omega_{mn}%
\end{array}
\right]  ,
\]
with $\omega_{pr}\geq0$, in terms of which the conditions
\[
\sum_{r=1}^{n}\omega_{pr}=\mu_{p}\text{ \ and \ }\sum_{p=1}^{m}\omega_{pr}%
=\nu_{r}%
\]
express the coupling properties of $\omega$. Note that in this finite context,
the transport plan $\omega$ says exactly that the portion $\omega_{pr}$ of the
probability $\mu_{p}$ of point $p$ in $\mathbf{A}$, is transported to point
$r$ in $\mathbf{B}$, with the coupling conditions above expressing the total
transport from $p$ and to $r$ respectively.

Then%
\[
\omega\left(  \alpha(a)\otimes b\right)  =\omega\left[
\begin{array}
[c]{c}%
\alpha(a)b_{1}\\
\vdots\\
\alpha(a)b_{n}%
\end{array}
\right]  =\sum_{p,q,r}\omega_{pr}\alpha_{pq}a_{q}b_{r},
\]
and similarly%
\[
\omega\left(  a\otimes\beta^{\prime}(b)\right)  =\sum_{q,r,s}\omega_{qs}%
a_{q}\beta_{sr}^{\prime}b_{r}.
\]
For $\omega$ to be a transport plan from $\mathbf{A}$ to $\mathbf{B}$, written
$\mathbf{A}\omega\mathbf{B}$ or $\omega\in T(\mathbf{A},\mathbf{B})$, requires
the further condition $\omega\left(  \alpha(a)\otimes b\right)  =\omega\left(
a\otimes\beta^{\prime}(b)\right)  $, expressed by%
\begin{equation}
\sum_{p=1}^{m}\omega_{pr}\alpha_{pq}=\sum_{s=1}^{n}\omega_{qs}\beta
_{sr}^{\prime}=\sum_{s=1}^{n}\omega_{qs}\frac{\nu_{r}}{\nu_{s}}\beta_{rs},
\label{KlasVervoerNaB}%
\end{equation}
which when summed over $r$ gives exactly (\ref{KlasVervoerNaPunt}). Hence
(\ref{KlasVervoerNaB}) is clearly a refinement of (\ref{KlasVervoerNaPunt}).

Essentially the probability $\mu_{p}$ of $p$ in $\mathbf{A}$ is split into $n$
``compartments'' having probabilities $\omega_{p1},...,\omega_{pn}$
respectively, with the transition probability $\alpha_{pq}$ to $q$ in
$\mathbf{A}$ then acting on each of these compartments. The sum of the
probabilities transported from the $r$'th compartment of each of $\mathbf{A}%
$'s points to $q$ by $\alpha$, is then required to be given by the total
transport by $\beta^{\prime}$ from $q$'s compartments to point $r$ in
$\mathbf{B}$.

Paraphrasing this: the transport from $r$ to $q$ by $\alpha$, is equal to the
transport from $q$ to $r$ by $\beta^{\prime}$. The analogy with detailed
balance is clear. Keep in mind that we are particularly interested in the case
where $\beta^{\prime}=\beta$, i.e., where $\mathbf{B}$ is in detailed balance.

In this way $\alpha$ is required in some measure to adhere to rules laid down
by $\beta$. In the one extreme case where $\omega=\mu\otimes\nu$, no
restrictions are placed on $\alpha$ by $\mathbf{A}\omega\mathbf{B}$. In the
other extreme, with $\left(  A,\mu\right)  =\left(  B,\nu\right)  $ and
$\omega$ the diagonal measure associated to $\mu$, i.e., $\omega_{pp}=\mu_{p}$
while $\omega_{pq}=0$ for $p\neq q$, the condition $\mathbf{A}\omega
\mathbf{B}$ forces $\alpha=\beta$. The role of a Wasserstein distance
$W(\mathbf{A},\mathbf{B})$ is to optimize this to determine how much of, or to
what extent, $\mathbf{B}$'s behaviour is reflected in that of $\mathbf{A}$.

Analogously for the transfer of probability from a single point $p$ in
$\mathbf{A}$ to the whole of $\mathbf{A}$, versus $\mathbf{B}\psi\mathbf{A}$.
The latter acts as a refinement of the dual property $\alpha(1_{m})=1_{m}$
with $1_{m}$ the column with $1$ as all $m$ entries, i.e., $\mu\alpha^{\prime
}=\mu$ (or $\mu^{\prime}\circ\alpha^{\prime}=\mu^{\prime}$ in the general
abstract notation), instead of $\mu\alpha=\mu$. As above, $\mathbf{B}%
\psi\mathbf{A}$ is expressed by
\begin{equation}
\sum_{r=1}^{n}\psi_{rp}\beta_{rs}=\sum_{q=1}^{m}\psi_{sq}\alpha_{qp}^{\prime
}=\sum_{q=1}^{m}\psi_{sq}\frac{\mu_{p}}{\mu_{q}}\alpha_{pq},
\label{KlasVervoerNaA}%
\end{equation}
the sum over $s$ of which gives the transfer of probability from a single
point $p$ in $\mathbf{A}$ to the whole of $\mathbf{A}$,%
\[
\sum_{q=1}^{m}\mu_{p}\alpha_{pq}=\mu_{p},
\]
which indeed simply expresses $\alpha(1_{m})=1_{m}$. In the paraphrased form
as above, (\ref{KlasVervoerNaA}) says that the transport from $p$ to $s$ by
$\beta$, is equal to the transport from $s$ to $p$ by $\alpha^{\prime}$. Again
this requires $\mathbf{A}$ to adhere in some degree to some aspects of
$\mathbf{B}$'s behaviour, the extent of which will be measured by the
Wasserstein distance $W(\mathbf{B},\mathbf{A})$.

To clarify the connection to the case $\mathbf{A}\omega\mathbf{B}$, we can
without loss of generality use the dual transport plan $\omega^{\prime}$,
which in this context (the general case to be treated in Section
\ref{AfdWEiensk}) is simply the reverse of a transport plan $\omega$ from
$\mathbf{A}$ to $\mathbf{B}$, namely
\[
\omega^{\prime}=\left[
\begin{array}
[c]{rrrcrrr}%
\omega_{11}^{\prime} & \cdots & \omega_{n1}^{\prime} & \cdots & \omega
_{1m}^{\prime} & \cdots & \omega_{nm}^{\prime}%
\end{array}
\right]  ,
\]
where $\omega_{rp}^{\prime}=\omega_{pr}$, with $\mathbf{B}\omega^{\prime
}\mathbf{A}$ therefore expressed by%
\[
\sum_{q=1}^{m}\omega_{qs}\frac{\mu_{p}}{\mu_{q}}\alpha_{pq}=\sum_{r=1}%
^{n}\omega_{pr}\beta_{rs}.
\]

Note that there is an asymmetry between $\mathbf{A}\omega\mathbf{B}$ and
$\mathbf{B}\omega^{\prime}\mathbf{A}$. This boils down to $\mathbf{A}%
\omega\mathbf{B}$ being equivalent to $\mathbf{B}^{\prime}\omega^{\prime
}\mathbf{A}^{\prime}$, as is easily checked, rather than to $\mathbf{B}%
\omega^{\prime}\mathbf{A}$. The set of transport plans $T(\mathbf{B}%
,\mathbf{A})$ need not consist of the reversals of the transport plans
$T(\mathbf{A},\mathbf{B})$. In particular, the optimal transport plan from
$\mathbf{B}$ to $\mathbf{A}$ need not be the reverse (i.e., the dual) of the
optimal transport plan from $\mathbf{A}$ to $\mathbf{B}$. Consequently, we can
typically expect $W(\mathbf{A},\mathbf{B})\neq W(\mathbf{B},\mathbf{A})$. If
both $\mathbf{A}$ and $\mathbf{B}$ satisfy detailed balance, then
$\mathbf{A}\omega\mathbf{B}$ and $\mathbf{B}\omega^{\prime}\mathbf{A}$ are
equivalent, which will ensure $W(\mathbf{A},\mathbf{B})=W(\mathbf{B}%
,\mathbf{A})$.




\subsection{An explicit example\label{OndAfdKlasVbUitdrGeval}}

In order to do explicit calculations easily, we consider an example for the
case $m=4$ and $n=2$. This example also illustrates how an appropriate
Wasserstein distance can be chosen, and the low dimensions allow us to find a
simple formula for it, which in turn sheds light on its role.

Let both the systems $\mathbf{A}$ and $\mathbf{B}$ be given by Markov chains
as in the previous subsection, where we take $\mathbf{A}$ to be the composite
of two 2-point systems, while $\mathbf{B}$ is a 2-point system. In the case of
$\mathbf{A}$, the indices in $\alpha_{pq}$ refer to the points of its 4-point
set $\{1,2\}\times\{1,2\}$ labeled as follows in terms of the original
2-point set $\{1,2\}$:
\[%
\begin{array}
[c]{cc}%
1\equiv(1,1), & 2\equiv(2,1)\\
3\equiv(1,2), & 4\equiv(2,2)
\end{array}
\]
For $\mathbf{B}$ on the other hand, we adapt the notation for the transition
matrix as follows:%
\[
\beta=\left[
\begin{array}
[c]{rr}%
\tilde{r} & r\\
s & \tilde{s}%
\end{array}
\right]
\]
with $0\leq r,s\leq1$, and
\[
\tilde{r}=1-r\text{ \ and \ }\tilde{s}=1-s.
\]
Recall that due to $\nu\beta=\nu$, $\mathbf{B}$ necessarily satisfies detailed
balance, namely
\[
\nu_{1}r=\nu_{2}s.
\]

We interpret $\mathbf{A}$ as a system consisting of two classical spins, and
$\mathbf{B}$ as a single classical spin. Natural ``coordinates'', measuring
these spin values are therefore%
\[
k_{1}=\frac{1}{2}\left[
\begin{array}
[c]{r}%
1\\
-1
\end{array}
\right]  \otimes1_{2}\text{ \ and \ }k_{2}=1_{2}\otimes\frac{1}{2}\left[
\begin{array}
[c]{r}%
1\\
-1
\end{array}
\right]
\]
for $\mathbf{A}$, where $1_{2}$ is the column with $1$ as both entries, and%
\begin{equation}
l=l_{1}=l_{2}=\frac{1}{2}\left[
\begin{array}
[c]{r}%
1\\
-1
\end{array}
\right]  \label{2puntKoord}%
\end{equation}
for $\mathbf{B}$. This will define the Wasserstein distances from $\mathbf{A}$
to $\mathbf{B}$, and from $\mathbf{B}$ to $\mathbf{A}$; see Subsection
\ref{OndAfdW}.

We assume that the cost $\omega c_{AB}$ for a transport plan $\omega\in
T(\mathbf{A},\mathbf{B})$ is given as in (\ref{eindDimKoste}) by the
\emph{cost matrix}
\[
c_{AB}=\left|  k_{1}\otimes1_{2}-1_{4}\otimes l\right|  ^{2}%
\]
from systems on $A$ to systems on $B$ (here $S_{\nu}l^{\ast}S_{\nu}=l$, since
$B$ is abelian), which in transposed form (for typographical convenience) is%
\[
c_{AB}^{\intercal}=(0,1,0,1,1,0,1,0).
\]
(In this subsection we often write row matrices as tuples to clearly delineate
the entries). This measures the square of the difference in the value of the
``first'' spin in $\mathbf{A}$ and the spin in $\mathbf{B}$. Lower cost should
therefore correspond to transitions from the set $\{1,3\}$ to the set
$\{2,4\}$ and vice versa in $\mathbf{A}$, in order to conform to transitions
between $1$ and $2$ in $\mathbf{B}$.

A convenient parametrization of the set $T(\mu,\nu)$ of transport plans
$\omega$ from $\mu$ to $\nu$ for this cost matrix, is%
\[
\omega=(\mu_{1}-\gamma_{1},\gamma_{2},\mu_{3}-\gamma_{3},\gamma_{4},\gamma
_{1},\mu_{2}-\gamma_{2},\gamma_{3},\mu_{4}-\gamma_{4}),
\]
where $0\leq\gamma_{i}\leq\mu_{i}$ for $i=1,2,3,4$, and
\[
\mu_{1}-\gamma_{1}+\gamma_{2}+\mu_{3}-\gamma_{3}+\gamma_{4}=\nu_{1}.
\]
This is convenient, since then the cost of the transport plan is then given by%
\[
\omega c_{AB}=\gamma_{1}+\gamma_{2}+\gamma_{3}+\gamma_{4}\text{,}%
\]
the minimum of which over all $\omega\in T(\mathbf{A},\mathbf{B})$ is the
squared Wasserstein distance $W(\mathbf{A},\mathbf{B})^{2}$.

Any $\omega\in T(\mathbf{A},\mathbf{B})$ is described by the further
conditions
\begin{align*}
(\mu_{1}\alpha_{12}+\mu_{3}\alpha_{32})+(-\gamma_{1}\alpha_{12}+\gamma
_{2}\alpha_{22}-\gamma_{3}\alpha_{32}+\gamma_{4}\alpha_{42})  &  =\mu
_{2}s+\gamma_{2}(\bar{r}-s)\\
(\mu_{1}\alpha_{14}+\mu_{3}\alpha_{34})+(-\gamma_{1}\alpha_{14}+\gamma
_{2}\alpha_{24}-\gamma_{3}\alpha_{34}+\gamma_{4}\alpha_{44})  &  =\mu
_{4}s+\gamma_{4}(\bar{r}-s)
\end{align*}
and
\begin{align*}
(\mu_{2}\alpha_{21}+\mu_{4}\alpha_{41})+(\gamma_{1}\alpha_{11}-\gamma
_{2}\alpha_{21}+\gamma_{3}\alpha_{31}-\gamma_{4}\alpha_{41})  &  =\mu
_{1}r+\gamma_{1}(\bar{r}-s)\\
(\mu_{2}\alpha_{23}+\mu_{4}\alpha_{43})+(\gamma_{1}\alpha_{13}-\gamma
_{2}\alpha_{23}+\gamma_{3}\alpha_{33}-\gamma_{4}\alpha_{43})  &  =\mu
_{3}r+\gamma_{3}(\bar{r}-s),
\end{align*}
which we now view as an extended or generalized detailed balance condition.
Due to the invariance $\mu\alpha=\mu$, there are other equivalent ways of
stating these conditions, however, the form above emphasizes transitions in
$\mathbf{A}$ from the set $\{1,3\}$ to the set $\{2,4\}$ in the first two
conditions, and vice versa in the last two. This generalized detailed balance
condition is trivial when $\omega=\mu\otimes\nu$, though, and is therefore not
of much use as it stands. We need to quantify how strong it is, which is where
the cost and Wasserstein distance come into play.

In particular, zero cost $\omega c_{AB}=0$ corresponds to the conditions%
\begin{equation}
\mu_{1}\alpha_{12}+\mu_{3}\alpha_{32}=\mu_{2}s\text{, \ }\mu_{1}\alpha
_{14}+\mu_{3}\alpha_{34}=\mu_{4}s \label{VbfbVwd1}%
\end{equation}
and%
\begin{equation}
\mu_{2}\alpha_{21}+\mu_{4}\alpha_{41}=\mu_{1}r\text{, \ }\mu_{2}\alpha
_{23}+\mu_{4}\alpha_{43}=\mu_{3}r, \label{VbfbVwd2}%
\end{equation}
which in our picture therefore give the form of the generalized detailed
balance condition in $\mathbf{A}$ closest (in terms of our chosen cost) to the
detailed balance in $\mathbf{B}$.
We note that these conditions are indeed satisfied by a multitude of systems
$\mathbf{A}$ and $\mathbf{B}$, hence zero cost can indeed be reached, in this
case for the uniquely determined transport plan%
\[
\omega=(\mu_{1},0,\mu_{3},0,0,\mu_{2},0,\mu_{4}).
\]
It is also clear that in this example $\mathbf{B}$ is uniquely determined by
$\mathbf{A}$ when $\omega c_{AB}=0$. In particular, $\mu_{1}+\mu_{3}=\nu_{1}$.

In the latter, we obviously have $W(\mathbf{A},\mathbf{B})=0$. However, more
generally for any $\mathbf{A}$ and $\mathbf{B}$ in this example, $W$ gives us
bounds on how far the generalized detailed balance condition is from the
optimal form (\ref{VbfbVwd1}) and (\ref{VbfbVwd2}). To see this, a simple
definition of the deviation from this optimal form can for example be taken
as
\begin{align*}
f  &  =\left|  \mu_{1}\alpha_{12}+\mu_{3}\alpha_{32}-\mu_{2}s\right|  +\left|
\mu_{1}\alpha_{14}+\mu_{3}\alpha_{34}-\mu_{4}s\right| \\
&  +\left|  \mu_{2}\alpha_{21}+\mu_{4}\alpha_{41}-\mu_{1}r\right|  +\left|
\mu_{2}\alpha_{23}+\mu_{4}\alpha_{43}-\mu_{3}r\right|  ,
\end{align*}
from which we clearly see that%
\begin{align*}
f  &  \leq4\left(  1+\left|  \tilde{r}-s\right|  \right)  W(\mathbf{A}%
,\mathbf{B})^{2}\\
&  =\left\{
\begin{array}
[c]{ll}%
4(\tilde{r}+\tilde{s})W(\mathbf{A},\mathbf{B})^{2} & \text{if }r+s\leq1\\
4(r+s)W(\mathbf{A},\mathbf{B})^{2} & \text{if }r+s\geq1.
\end{array}
\right.
\end{align*}
For the special case%
\[
r+s=1
\]
this bound takes the simple form%
\[
f\leq4W(\mathbf{A},\mathbf{B})^{2}.
\]
Because of the signs in front of the $\gamma_{i}$'s in the conditions for
$\omega\in T(\mathbf{A},\mathbf{B})$ above, this is not necessarily a very
tight bound, though. A version of all this for each of the individual terms in
$f$ can also be written down.



Perfectly analogous results are obtained for transport from $\mathbf{B}$ to
$\mathbf{A}$, in terms of $\alpha^{\prime}$ instead of $\alpha$, and
$W(\mathbf{B},\mathbf{A})$ instead of $W(\mathbf{A},\mathbf{B})$. We note that
$\mathbf{B} \omega' \mathbf{A}$ 
can also be expressed via 
$\mathbf{A}' \omega \mathbf{B}'$ 
(see Section \ref{AfdWEiensk}).

To be clear, the cost matrix from $B$ to $A$ is obtaining by simply swapping
the roles of $A$ and $B$. That is,
\[
c_{BA}=\left|  l\otimes1_{4}-1_{2}\otimes k_{1}\right|  ^{2},
\]
while the parametrization of the set $T(\nu,\mu)$ of transport plans $\omega$
from $\nu$ to $\mu$ becomes%
\[
\omega^{\prime}=(\mu_{1}-\gamma_{1},\gamma_{1},\gamma_{2},\mu_{2}-\gamma
_{2},\mu_{3}-\gamma_{3},\gamma_{3},\gamma_{4},\mu_{4}-\gamma_{4}).
\]
We then have equal cost in the two directions, $\omega^{\prime}c_{BA}=\omega
c_{AB}$.

For zero transport cost from $\mathbf{A}$ to $\mathbf{B}$, the condition
$\mathbf{A}\omega\mathbf{B}$ is expressed by (\ref{VbfbVwd1}) and
(\ref{VbfbVwd2}).
On the other hand, for zero transport cost from $\mathbf{B}$ to $\mathbf{A}$,
the condition $\mathbf{B}\omega^{\prime}\mathbf{A}$ is expressed by
\begin{equation}
\alpha_{12}+\alpha_{14}=\alpha_{32}+\alpha_{34}=r\text{ \ and \ }\alpha
_{21}+\alpha_{23}=\alpha_{41}+\alpha_{43}=s \label{Vb_Bw'A}%
\end{equation}
That is, the transport conditions $\mathbf{A}\omega\mathbf{B}$ and
$\mathbf{B}\omega^{\prime}\mathbf{A}$ place different restrictions on
$\mathbf{A}$, which means that $W(\mathbf{A},\mathbf{B})$ and $W(\mathbf{B}%
,\mathbf{A})$ are in general not simultaneously zero:
for given $(A,\mu)$ and $\mathbf{B}$, the values $W(\mathbf{A},\mathbf{B})=0$
and $W(\mathbf{B},\mathbf{A})=0$ are respectively reached on different sets of
$\alpha$'s. In other words, there are $\mathbf{A}$ and $\mathbf{B}$ such that
\begin{equation}
W(\mathbf{A},\mathbf{B})=0\text{ \ and \ }W(\mathbf{B},\mathbf{A})\neq0,
\label{asimVbDeel1}%
\end{equation}
as well as $\mathbf{A}$ and $\mathbf{B}$ such that%
\begin{equation}
W(\mathbf{A},\mathbf{B})\neq0\text{ \ and \ }W(\mathbf{B},\mathbf{A})=0.
\label{asimVbDeel2}%
\end{equation}
For example,
in the case $\alpha_{pq}>0$ for all $p$ and $q$, transform $\alpha$ such that
$\alpha_{11}\mapsto\alpha_{11}+\varepsilon$, $\alpha_{13}\mapsto\alpha
_{13}-\varepsilon$, $\mu_{2}\alpha_{21}\mapsto\mu_{2}\alpha_{21}-\mu
_{1}\varepsilon$ and $\mu_{2}\alpha_{23}\mapsto\mu_{2}\alpha_{23}+\mu
_{1}\varepsilon$, while all the other $\alpha_{pq}$'s are left unchanged.
Then, for small enough $\left|  \varepsilon\right|  $, (\ref{oorgWaarskDef}),
(\ref{KlasVervoerNaPunt}) and (\ref{Vb_Bw'A}) are preserved, but
(\ref{VbfbVwd2}) not. So even if $\mathbf{A}\omega\mathbf{B}$ and
$\mathbf{B}\omega^{\prime}\mathbf{A}$ both hold initially, then after this
transformation only $\mathbf{B}\omega^{\prime}\mathbf{A}$ continues to hold.
In particular, in the case of the unique couplings $\omega$ and $\omega
^{\prime}$ (determined by $\gamma_{1}=...=\gamma_{4}=0$) which can
respectively lead to $W(\mathbf{A},\mathbf{B})=0$ and $W(\mathbf{B}%
,\mathbf{A})=0$, only the latter still holds after the transformation.

The cost matrix
\[
c_{AB}=\left|  k_{1}\otimes1_{2}-1_{4}\otimes l\right|  ^{2}+\left|
k_{2}\otimes1_{2}-1_{4}\otimes l\right|  ^{2}%
\]
can similarly be studied, and as may be expected, leads to a weak form of the
usual detailed balance condition $\mu_{1}\alpha_{14}=\mu_{4}\alpha_{41}$
between the point $(1,1)$ and $(2,2)$ in $\mathbf{A}$.

\section{Quantum detailed balance and KMS duals\label{AfdFB&KMS}}

In this section we discuss some points around quantum detailed balance and a
related type of dual of u.p. maps, namely KMS-duals. The latter play a basic
role in Sections \ref{AfdW} and \ref{AfdWEiensk}, while the former provides
background to how we connect detailed balance and Wasserstein metrics in those
two sections. As in previous sections we only consider discrete time. To
extend to continuous time, one merely has to replace $\alpha$ by $\alpha_{t}$
for all $t\geq0$, throughout. Sections \ref{AfdW} and \ref{AfdWEiensk} will
handle this more generally and systematically.

We are interested in defining detailed balance conditions and extensions
thereof for a quantum system $\mathbf{A}$. As was seen in the previous
section, for a classical Markov chain $\mathbf{A}$ the dual system
$\mathbf{A}^{\prime}$ is exactly the time-reversal of the dynamics in the
classical sense. Detailed balance of a classical Markov chain $\mathbf{A}$ is
therefore given by%
\begin{equation}
\mathbf{A}^{\prime}=\mathbf{A}. \label{AbstrakteKlasFB}%
\end{equation}
However, in the quantum or abstract noncommutative setting, this condition no
longer makes sense, as $\alpha^{\prime}:A^{\prime}\rightarrow A^{\prime}$ as
opposed to $\alpha:A\rightarrow A$, with $A^{\prime}\neq A$. Therefore a
direct mathematical translation of (\ref{AbstrakteKlasFB}) to the quantum case
does not appear possible. Consequently one has to explore other options to
define quantum detailed balance. This has lead to a range of possibilities
being proposed and studied over the past five decades. This includes \cite{Ag,
Al, An, CW, DF, DS1, FR, FU, KFGV, M, MS, RLC}, to mention only a few of the
papers. Although our main concern is to develop a general approach which can
ultimately extend such conditions into the non-equilibrium realm, we
nevertheless hope that our approach may also shed some light on possible
alternative quantum detailed balance conditions and their meaning. Indeed,
there is still some uncertainty about the best way to approach quantum
detailed balance, as has been pointed out in \cite{Diehl15, RLC}

There is in fact an alternative dual, the KMS-dual, of $\alpha$, denoted by
$\alpha^{\sigma}:A\rightarrow A$, which does map from $A$ to $A$, giving us
the system $\mathbf{A}=(A,\alpha^{\sigma},\mu)$. However, that does not
necessarily mean that $\mathbf{A}^{\sigma}=\mathbf{A}$ is a physically
sensible definition of quantum detailed balance; also see Example
\ref{sfb-vs-fb}.
Nevertheless, the KMS dual does appear in what are known as standard quantum
detailed balance conditions, one of which will be defined shortly.

\begin{definition}
\label{KMS-duaal}Given a u.p. map $E:A\rightarrow B$ satisfying $\nu\circ
E=\mu$ for some $\mu\in\mathfrak{F}(A)$ and $\nu\in\mathfrak{F}(B)$, we define
its \emph{KMS-dual} (w.r.t. $\mu$ and $\nu$) as
\[
E^{\sigma}:=j_{A}\circ E^{\prime}\circ j_{B}:B\rightarrow A.
\]
\end{definition}

Correspondingly $E^{\sigma}$ is u.p. and $\mu\circ E^{\sigma}=\nu$, while
$(E^{\sigma})^{\sigma}=E$ follows from $(E^{\prime})^{\prime}=E$. Under the
given assumptions, $E^{\sigma}$ is necessarily normal. If $E$ is u.c.p., then
so is $E^{\sigma}$. We note that KMS-duals have been discussed and used in
\cite{Pet, OPet, GL, C, FR, DS2}.

If $A$ is abelian, $j_{A}$ is the identity map, hence
\[
\alpha^{\sigma}=\alpha^{\prime},
\]
which indicates that at least in some ways the KMS dual is a natural
mathematical extension of time-reversal of the classical Markov chain in
Subsection \ref{OndAfdKlasVbOpset} to the general noncommutative framework.

The KMS-dual allows us to define a good example of a quantum detailed
property, which we can use in subsequent sections to illustrate certain points
about the general approach of this paper. This is done via a \emph{reversing
operation} $\theta$ for $(A,\mu)$, which is a $\ast$-anti-homomorphism
$\theta:A\rightarrow A$ (i.e., $\theta$ is linear, $\theta(a^{\ast}%
)=\theta(a)^{\ast}$, and $\theta(a_{1}a_{2})=\theta(a_{2})\theta(a_{1})$) such
that $\theta^{2}=\operatorname{id}_{A}$ and $\mu\circ\theta=\mu$. We then
consider the \emph{reverse} of $\mathbf{A}$ (with respect to $\theta$),
namely
\[
\mathbf{A}^{\leftarrow}:=(A,\alpha^{\leftarrow},\mu),
\]
with
\[
\alpha^{\leftarrow}:=\theta\circ\alpha^{\sigma}\circ\theta
\]
called the $\theta$\emph{-KMS dual} of $\alpha$ (see \cite{BQ2}). The latter
is also denoted by $\alpha^{\theta}$, but the reversing operation will always
be clear from context. It is easily checked that $\alpha^{\leftarrow
\leftarrow}\equiv(\alpha^{\leftarrow})^{\leftarrow}=\alpha$ since
$\theta^{\sigma}=\theta$; see \cite[Proposition 6.4]{DS2}. Hence
$\mathbf{A}^{\leftarrow\leftarrow}=\mathbf{A}$.

Now we define \emph{standard quantum detailed balance with respect to the
reversing operation} $\theta$ of $\mathbf{A}$, abbreviated as $\theta
$\emph{-sqdb}, by requiring that
\[
\mathbf{A}^{\leftarrow}=\mathbf{A}.
\]

\begin{example}
\label{sfb-vs-fb}Consider a system 
$\mathbf{A} = (M_{n},\alpha,\mu)$ 
on $M_n$
with $\alpha$ and $\mu$ represented directly on $M_{n}$ instead of on the
standard form discussed in Example \ref{VbEindDimStandVorm}. Then
\[
\alpha^{\sigma}(a)=\alpha^{\prime}(a^{\intercal})^{\intercal},
\]
with both these duals expressed in this representation, simply because
$j_{M_{n}\otimes1_{n}}$ is given by transposition. Assuming that $\mu$ is
given by a diagonal density matrix, and taking the reversing operation
$\theta$ on $M_{n}$ to be transposition (this is a standard choice), we
consequently have $\alpha^{\leftarrow}=\alpha^{\prime}$. In this case $\theta
$-sqdb is therefore expressed by $\alpha^{\prime}=\alpha$, i.e., by
\[
\mathbf{A}^{\prime}=\mathbf{A},
\]
where
$\mathbf{A}' =(M_{n}, \alpha', \mu)$,
emphasizing the similarity of $\theta$-sqdb to
the classical case (\ref{AbstrakteKlasFB}). The condition $\mathbf{A}^{\sigma
}=\mathbf{A}$ instead says $\theta\circ\alpha^{\prime}\circ\theta=\alpha$.
Essentially $j_{A}$ and $\theta$ cancel in $\theta$-sqdb.
\end{example}

For more on $\theta$\emph{-sqdb}, refer to \cite{BQ2, FR, FU}. Here we simply
note that $\theta$-sqdb can be expressed in terms of transport plans (or
balance) as set out in \cite[Corollary 6.7]{DS2}. The same is in fact true of
usual detailed balance of a Markov chain, as was discussed for example in
\cite[Section II]{Ver} and \cite[Section 3]{DS1}, though from a conceptually
different point of view and without connecting it to ideas from optimal transport.

%
%
%

\section{A simple quantum example\label{AfdKwantVb}}

To get an impression of our framework for quantum systems, specifically in
relation to detailed balance conditions, we take a brief look at a Wasserstein
distance between a very simple spin 1/2 system and a 2-point classical system.

Let $\mathbf{B}$ be as in Subsection \ref{OndAfdKlasVbUitdrGeval}. In this
section the quantum system $\mathbf{A}=(A,\alpha,\mu)$ will be defined
directly on $M_{2}$, instead of in standard form (also see Example
\ref{sfb-vs-fb}), with
\[
A=M_{2}\text{, \ }\alpha(a)=\lambda U_{\eta}^{\ast}aU_{\eta}+\tilde{\lambda
}U_{\varphi}^{\ast}aU_{\varphi}\text{ \ and \ }\mu(a)=\operatorname*{Tr}%
(\rho_{\mu}a)
\]
for all $a\in A$, in terms of the unitary and density matrices
\[
U_{\eta}=\left[
\begin{array}
[c]{rr}%
1 & 0\\
0 & e^{i\eta}%
\end{array}
\right]  \text{ \ and \ }\rho_{\mu}=\left[
\begin{array}
[c]{rr}%
\mu_{1} & 0\\
0 & \mu_{2}%
\end{array}
\right]
\]
respectively, with $\eta,\varphi\in\mathbb{R}$, $0\leq\lambda\leq1$ and
$\tilde{\lambda}=1-\lambda$. It is then elementary, though somewhat tedious,
to deduce the results below.

With transposition in the given basis serving as a reversing operation
$\theta$, $\mathbf{A}$ satisfies $\theta$-sqdb (Example \ref{sfb-vs-fb} gives
a definition sufficient for this example) for any values of 
$\eta$, $\varphi$ and $\lambda$. 
However, depending on certain Wasserstein distances from
$\mathbf{A}$ to $\mathbf{B}$, the former can have an additional property,
which may also be interpreted as a form of detailed balance, as will be
discussed below.

Taking the coordinates (as in Subsection \ref{OndAfdW}) for $\mathbf{A}$ as
the Pauli spin matrices,%
\[
k_{1}=\frac{1}{2}\left[
\begin{array}
[c]{rr}%
1 & 0\\
0 & -1
\end{array}
\right]  ,k_{2}=\frac{1}{2}\left[
\begin{array}
[c]{rr}%
0 & 1\\
1 & 0
\end{array}
\right]  ,k_{3}=\frac{1}{2}\left[
\begin{array}
[c]{rr}%
0 & -i\\
i & 0
\end{array}
\right]  ,
\]
and again using (\ref{2puntKoord}) for $\mathbf{B}$, also setting $l_{3}=l$,
then via (\ref{eindDimKoste}) and (\ref{spesW}) for $d=3$, the condition
$W(\mathbf{A},\mathbf{B})<1$ implies that $r+s>0$ and%
\begin{equation}
\lambda\sin\eta+\tilde{\lambda}\sin\varphi=0\text{.} \label{verfyndeFBVwd}%
\end{equation}
In fact, if $\mathbf{A}\omega\mathbf{B}$ for $\omega\neq\mu\otimes\nu$, and we
assume $r+s>0$, then (\ref{verfyndeFBVwd}) follows. However, $W(\mathbf{A}%
,\mathbf{B})<1$ implies all of this, and this inequality can indeed be
attained on certain sets of systems $\mathbf{A}$ and $\mathbf{B}$.

To see the significance of (\ref{verfyndeFBVwd}), consider any density
matrices $X$ and $Y$ for $\mathbf{A}$. Let
\[
\alpha^{\ast}(X)=\lambda U_{\eta}XU_{\eta}^{\ast}+\tilde{\lambda}U_{\varphi
}XU_{\varphi}^{\ast}%
\]
express the action of $\alpha$ in terms of $X$. Set
\[
V_{XY}=\operatorname*{Tr}(\alpha^{\ast}(X)Y)-\operatorname*{Tr}(XY).
\]
One then finds that $V_{XY}=V_{YX}$ if and only if (\ref{verfyndeFBVwd}) or
$\operatorname{Im}(X_{21}Y_{12})=0$ holds, where the latter refers to entries
of the matrices $X$ and $Y$.

For pure states $X$ and $Y$, this has a simple interpretation: If the physical
system is currently in the state $X$ (with $\mu$ now viewed as a ``reference''
state, rather than necessarily being the state in which the system finds
itself), then $\operatorname*{Tr}(XY)$ is the probability for the system being
found to be in state $Y$ when measuring an observable with $Y$ as an
eigenstate at the current time. Correspondingly for $\operatorname*{Tr}%
(\alpha^{\ast}(X)Y)$, but one step into the future, when the system is in
state $\alpha^{\ast}(X)$. Hence we can heuristically think of $V_{XY}$ as a
flow of probability from $X$ to $Y$ during one step of time evolution of the
state $X$. The condition $V_{XY}=V_{YX}$ tells us that this flow is the same
in both directions, i.e., it can be viewed as a form of detailed balance, at
least between pure states $X$ and $Y$ for which $\operatorname{Im}%
(X_{21}Y_{12})\neq0$. The condition $\operatorname{Im}(X_{21}Y_{12})=0$, on
the other hand, describes a small set in the Cartesian product of the set of
spin 1/2 pure states with itself, in the sense that the condition can cease to
hold due to an arbitrarily small change in either of the states $X$ or $Y$.

\section{Wasserstein distances\label{AfdW}}


We now have the framework in place and our basic approach set out, motivated,
and illustrated by simple examples. The next step is to systematically develop
the quantitative tool, namely Wasserstein distance from one system to another,
as well as related results, to provide the mathematical foundation for our
approach. So far we only defined and used Wasserstein distances in special
situations and examples. Here we develop the more general theory that will in
particular also cover those cases. This extends the definitions and some
results on metric properties from \cite{D1, D2} (where all systems were on the
same algebra).

The outlines of the arguments in this section are similar to those papers, but
need to be adapted. The arguments in the next section, studying aspects of
Wasserstein distances relevant to the properties and structure of systems,
move further afield, however. We use the notation and conventions from Section
\ref{AfdBasKon}, though the definition of a system will be extended and refined.

In Section \ref{AfdBasKon} we only considered a simple case of the systems we
are interested in. We now give a more abstract definition of systems,
including the ``coordinates'' from Subsection \ref{OndAfdW}, and allowing more
general cases than just discrete time. We assume the minimum structure which
still allows us to define Wasserstein distances from one system to another. In
particular, we do not require semigroup properties for dynamics, which allows
for non-Markovian systems. The definition is taken from \cite[Section 3]{D2},
though generalized to positive dynamics instead of just completely positive
dynamics, for the reasons mentioned in Subsection \ref{OndAfdDef}, which
become relevant in the next section.

\begin{definition}
\label{stelsel}A \emph{system} is given by $\mathbf{A}=\left(  A,\alpha
,\mu,k\right)  $, where $\mu\in\mathfrak{F}(A)$ and $k=(k_{1},...,k_{d})$ with
$k_{1},...,k_{d}\in A$ for some $d\in\{1,2,3,...\}$, while $\alpha$ consists
of the following: Let $\Upsilon$ be any set. To each $\upsilon\in\Upsilon$
corresponds a set $Z_{\upsilon}$ and \emph{generalized dynamics}
$\alpha_{\upsilon}$ on $A$, which is given by a u.p. map
\[
\alpha_{\upsilon,z}:A\rightarrow A
\]
for every $z\in Z_{\upsilon}$, such that
\[
\mu\circ\alpha_{\upsilon,z}=\mu
\]
for all $z\in Z_{\upsilon}$ and $\upsilon\in\Upsilon$. We then write
\[
\alpha=\left(  \alpha_{\upsilon}\right)  _{\upsilon\in\Upsilon}.
\]
\end{definition}

We view $Z_{\upsilon}$ as a set of ``points in time'' in an abstract sense.
Each $\alpha_{\upsilon}$ is viewed as dynamics, so $\mathbf{A}$ is really a
set of dynamical systems on $A$, one for each $\upsilon$, and all leaving the
state $\mu$ invariant. Lastly, $k$ serves as ``coordinates'' for the system,
in terms of which a quadratic cost of transport from one system to another
will be defined. We also define $\alpha^{\prime}$ and $\alpha^{\sigma}$
through
\[
\alpha_{\upsilon,z}^{\prime}=(\alpha_{\upsilon,z})^{\prime}\text{ \ \ and
\ \ }\alpha_{\upsilon,z}^{\sigma}=(\alpha_{\upsilon,z})^{\sigma}.
\]

In terms of Definition \ref{KMS-duaal} we have the following type of dual of a
system, which will play a key role in one of the types of Wasserstein
distances to be defined shortly.

\begin{definition}
\label{stelselKMS-duaal}The \emph{KMS-dual} of $\mathbf{A}$ is defined to be
$\mathbf{A}^{\sigma}:=(A,\alpha^{\sigma},\mu,k)$.
\end{definition}

Note that $\mathbf{A}^{\sigma\sigma}=\mathbf{A}$. Extensions of the duals of
$\mathbf{A}$ in Subsection \ref{OndAfdDef} and Section \ref{AfdFB&KMS} will be
defined in the next section to study corresponding symmetries of Wasserstein distances.

For the remainder of this section and the next, we fix $\Upsilon$, the
$Z_{\upsilon}$'s, and $d$, i.e., they are the same for all systems our Wasserstein distances will be applied to. The
following notational convention will be used: $\mathbf{A}$ will denote
$\left(  A,\alpha,\mu,k\right)  $, as in the definition above, and similarly
we write
\[
\mathbf{B}=\left(  B,\beta,\nu,l\right)  \text{ \ \ and \ \ }\mathbf{C}%
=\left(  C,\gamma,\xi,m\right)  .
\]
Let
\[
X
\]
denote any set of such systems. Our Wasserstein distances will be defined on
$X$.

Notationally it will be convenient to allow systems with one point $\Upsilon$ and without
coordinates (the case $d=0$) as well, written as $\left(  A,\alpha,\mu\right)
$, though they will not be viewed as members of $X$. This is needed to
formulate certain conditions related to transport plans, involving modular
dynamics, more conveniently in the definition below.

We are going to define Wasserstein distances from one system $\mathbf{A}$ to
another $\mathbf{B}$, as the square root of the infimum of the cost over an
appropriate set of transport plans from $\mathbf{A}$ to $\mathbf{B}$. One
option is the full set $T(\mathbf{A},\mathbf{B})$ of transport plans as
before, but there are further refinements that can be made to the set of
transport plans, that can for example ensure symmetry of the resulting
Wasserstein distance.

The sets of transport plans relevant to us are defined as follows, refining
\cite[Definition 3.5]{D2} with a corresponding change in notation and terminology.

\begin{definition}
\label{T(A,B)}We write
\[
\mathbf{A}\omega\mathbf{B}%
\]
to express that $\omega\in T(\mu,\nu)$ (see Subsection \ref{OndAfdOP}) and
\begin{equation}
\omega(\alpha_{\upsilon,z}(a)\otimes b^{\prime})=\omega(a\otimes
\beta_{\upsilon,z}^{\prime}(b^{\prime})) \label{balItvKop}%
\end{equation}
for all $a\in A$, $b^{\prime}\in B^{\prime}$, $z\in Z_{\upsilon}$ and
$\upsilon\in\Upsilon$, in which case we call $\omega$ a \emph{transport plan}
from $\mathbf{A}$ to $\mathbf{B}$. The set of transport plans from
$\mathbf{A}$ to $\mathbf{B}$ is written
\[
T(\mathbf{A},\mathbf{B}):=\left\{  \omega\in T(\mu,\nu):\mathbf{A}%
\omega\mathbf{B}\right\}  .
\]
The set of \emph{modular} transport plans from $\mathbf{A}$ to $\mathbf{B}$
is
\[
T_{\sigma}(\mathbf{A},\mathbf{B}):=\left\{  \omega\in T(\mathbf{A}%
,\mathbf{B}):(A,\sigma^{\mu},\mu)\omega(B,\sigma^{\nu},\nu)\right\}  .
\]
where $\sigma^{\mu}$ and $\sigma^{\nu}$ denote the modular dynamics (from
Tomita-Takesaki theory) associated to $\mu$ and $\nu$ respectively. The set of
\emph{KMS} transport plans from $\mathbf{A}$ to $\mathbf{B}$ is%
\[
T_{\sigma\sigma}(\mathbf{A},\mathbf{B}):=\left\{  \omega\in T_{\sigma
}(\mathbf{A},\mathbf{B}):\mathbf{A}^{\sigma}\omega\mathbf{B}^{\sigma}\right\}
.
\]
\end{definition}

In connection to $T_{\sigma}(\mathbf{A},\mathbf{B})$ we point out that
$(\sigma_{t}^{\nu})^{\prime}=\sigma_{t}^{\nu^{\prime}}$.
Note that the transport plans do not depend on the coordinates $k$ and $l$.
Furthermore, we always have $\mu\odot\nu^{\prime}\in T_{\sigma\sigma
}(\mathbf{A},\mathbf{B})$.

\begin{example}
\label{EindDimModGroep}Returning to our special case in Example
\ref{VbEindDimStandVorm}, the modular dynamics (or modular group) $\sigma
^{\mu}$ is given by%
\[
\sigma_{t}^{\mu}(\pi(a))=\pi(\rho_{\mu}^{it}a\rho_{\mu}^{-it})
\]
for all $a\in M_{n}$. Along with the modular conjugation, the modular group is
a natural component of the Tomita-Takesaki (or modular) theory of von Neumann
algebras; see \cite{BR1, T2} for expositions of the general theory.
\end{example}

Next we need to formulate the cost of a transport plan. In Subsection
\ref{OndAfdOP} we mentioned that a transport plan $\omega\in T(\mu,\nu)$ from
$\mu$ to $\nu$ can be represented as a channel (u.c.p. map) $E_{\omega
}:A\rightarrow B$. To see this, we first define a state $\delta_{\nu}$ on
$B\odot B^{\prime}$ by requiring that%
\[
\delta_{\nu}(b\otimes b^{\prime})=\left\langle \Lambda_{\nu},bb^{\prime
}\Lambda_{\nu}\right\rangle
\]
for all $b\in B$ and $b^{\prime}\in B^{\prime}$.

\begin{definition}
\label{KanVsToest}For any $\omega\in T(\mu,\nu)$, the channel $E_{\omega}$ is
defined as the unique function from $A$ to $B$ such that%
\begin{equation}
\omega(a\otimes b^{\prime})=\delta_{\nu}(E_{\omega}(a)\otimes b^{\prime}).
\label{omegaItv_delta&Eomega}%
\end{equation}
\end{definition}

This provides a correspondence between elements of $T(\mu,\nu)$ and channels
$E:A\rightarrow B$ satisfying $\nu\circ E=\mu$. The map $E_{\omega}$ in this
condition necessarily exists and is indeed forced to be u.c.p. All this is
demonstrated in \cite[Section 3]{DS2}. The state $\delta_{\nu}$ generalizes
the so-called diagonal state associated to a probability measure, but can also
be viewed as an abstract version of the maximally entangled state of a quantum
state with itself. In this way Definition \ref{KanVsToest} provides a
generalization of the (finite dimensional) Choi-Jamio{\l}kowski duality
between channels and states as used in quantum physics \cite{Choi, J, GKM}, to
a von Neumann algebra framework.

For our purposes, $E_{\omega}$ has two uses. Firstly, we have $\mathbf{A}%
\omega\mathbf{B}$ if and only if $\omega\in T(\mu,\nu)$ and%
\begin{equation}
E_{\omega}\circ\alpha_{\upsilon,z}=\beta_{\upsilon,z}\circ E_{\omega}
\label{BalE}%
\end{equation}
for all $z\in Z_{\upsilon}$ and $\upsilon\in\Upsilon$, by \cite[Theorem
4.1]{DS2}, which we also write simply as $E_{\omega}\circ\alpha=\beta\circ
E_{\omega}$. This is often a convenient way to verify $\mathbf{A}%
\omega\mathbf{B}$.

Secondly, $E_{\omega}$ allows us to define transport cost while avoiding the
complications mention in Subsection \ref{OndAfdW}, due to unbounded $S_{\nu}$:

\begin{definition}
\label{I}The cost of a transport plan $\omega\in T(\mathbf{A},\mathbf{B})$ is
defined as%
\begin{equation}
I_{\mathbf{A},\mathbf{B}}(\omega)=\sum_{i=1}^{d}\left[  \mu(k_{i}^{\ast}%
k_{i})+\nu(l_{i}^{\ast}l_{i})-\nu(E_{\omega}(k_{i})^{\ast}l_{i})-\nu
(l_{i}^{\ast}E_{\omega}(k_{i}))\right]  . \label{algKoste}%
\end{equation}
\end{definition}

In general $S_{\nu}b^{\ast}S_{\nu}$ is affiliated with $B^{\prime}$, but it
may be unbounded, hence not contained in $B^{\prime}$. For cases where we do
have $S_{\nu}b^{\ast}S_{\nu}\in B^{\prime}$, the form of the cost above is in
fact equivalent to the form used in (\ref{eindDimKoste}).
But unlike the latter, (\ref{algKoste}) only refers to bounded operators, thus
avoiding complications even when some $S_{\nu}l_{i}^{\ast}S_{\nu}$'s are unbounded.

\begin{remark}
\label{Ialt}In terms of representation machinery built and used in
\cite[Section 2 and 3]{D1}, but which will not be reviewed in any detail here,
one can write
\[
I_{\mathbf{A},\mathbf{B}}(\omega)=\left\|  \pi_{\mu}^{\omega}(k)\Omega
-\pi_{\nu}^{\omega}(l)\Omega\right\|  _{\oplus\omega}^{2},
\]
in terms of cyclic representations $(H_{\mu}^{\omega},\pi_{\mu}^{\omega
},\Omega)$ and $(H_{\nu}^{\omega},\pi_{\nu}^{\omega},\Omega)$ of $\left(
A,\mu\right)  $ and $\left(  B,\nu\right)  $ inherited from a cyclic
representation $(H_{\omega},\pi_{\omega},\Omega)$ of $(A\odot B^{\prime
},\omega)$, and where $\pi_{\mu}^{\omega}(k)\Omega\equiv\left(  \pi_{\mu
}^{\omega}(k_{1})\Omega,...,\pi_{\mu}^{\omega}(k_{d})\Omega\right)
\in\bigoplus_{l=1}^{n}H_{\omega}$, etc. This starts to give a more intuitive
idea of why we expect to obtain distances from this cost. Indeed, the triangle
inequality of our Wasserstein distances below follow from this form of
$I_{\mathbf{A},\mathbf{B}}(\omega)$. On the other hand, (\ref{algKoste}) is
useful in relation to symmetry, as well as the consequences of zero
Wasserstein distance in Section \ref{AfdWEiensk}.
\end{remark}

With this transport cost in hand, we can define our Wasserstein distances from
one system to another. The cost $I_{\mathbf{A},\mathbf{B}}(\omega)$ does not
depend directly on $\alpha$ and $\beta$, but they play a central role in
determining the allowed transport plans, namely $T(\mathbf{A},\mathbf{B})$,
$T_{\sigma}(\mathbf{A},\mathbf{B})$ or $T_{\sigma\sigma}(\mathbf{A}%
,\mathbf{B})$, and enter the definition of Wasserstein distance via this.

\begin{definition}
\label{W2}Given a set $X$ of systems as above, we define the \emph{Wasserstein
distance} $W$\ on $X$ by
\[
W(\mathbf{A},\mathbf{B}):=\inf_{\omega\in T(\mathbf{A},\mathbf{B}%
)}I_{\mathbf{A},\mathbf{B}}(\omega)^{1/2},
\]
the \emph{modular Wasserstein distance} $W_{\sigma}$\ on $X$ by
\[
W_{\sigma}(\mathbf{A},\mathbf{B}):=\inf_{\omega\in T_{\sigma}(\mathbf{A}%
,\mathbf{B})}I_{\mathbf{A},\mathbf{B}}(\omega)^{1/2},
\]
and the \emph{KMS Wasserstein distance} $W_{\sigma}$\ on $X$ by\emph{ }
\[
W_{\sigma\sigma}(\mathbf{A},\mathbf{B}):=\inf_{\omega\in T_{\sigma\sigma
}(\mathbf{A},\mathbf{B})}I_{\mathbf{A},\mathbf{B}}(\omega)^{1/2},
\]
for all $\mathbf{A},\mathbf{B}\in X$, in terms of Definitions \ref{T(A,B)} and
\ref{I}.
\end{definition}

Note that we thus obviously have%
\[
W(\mathbf{A},\mathbf{B})\leq W_{\sigma}(\mathbf{A},\mathbf{B})\leq
W_{\sigma\sigma}(\mathbf{A},\mathbf{B})
\]
for all $\mathbf{A},\mathbf{B}\in X$.

We proceed to prove the existence of optimal transport plans and the main
metric properties of the functions $W$, $W_{\sigma}$ and $W_{\sigma\sigma}$.
The proofs are similar to those appearing in \cite{D1} and \cite[Section
3]{D2} for systems on the same von Neumann algebra. We simply show how to
adapt them. Much of the framework of \cite{D1} was in fact already set up in
terms of multiple von Neumann algebras, although the case of distance
functions between states on different algebras was not treated there.

\begin{definition}
An \emph{optimal} transport plan for $W_{\sigma\sigma}(\mathbf{A},\mathbf{B})$
is any $\omega\in T_{\sigma\sigma}(\mathbf{A},\mathbf{B})$ such that
$I_{\mathbf{A},\mathbf{B}}(\omega)^{1/2}=W_{\sigma\sigma}(\mathbf{A}%
,\mathbf{B})$. We say that optimal transport plans for $W_{\sigma\sigma}$
\emph{always exist} if they exist for $W_{\sigma\sigma}(\mathbf{A}%
,\mathbf{B})$ for all $\mathbf{A},\mathbf{B}\in X$. Correspondingly for $W$
and $W_{\sigma}$.
\end{definition}

Recall that for any set $Y$, a real-valued function $\rho$ on $Y\times Y$ is
called an \emph{asymmetric pseudometric} if it satisfies the triangle
inequality, $\rho(x,y)\geq0$ and $\rho(x,x)=0$ for all $x,y\in X$. If in
addition $\rho(x,y)=\rho(y,x)$, we call $\rho$ a \emph{pseudometric}. These
properties of our Wasserstein distances will be proven below. The remaining
standard property, namely that $\rho(x,x)=0$ implies $x=x$, will appear in a
more refined form in the next section, and is essentially recovered in its
usual form in the case of appropriate coordinates (the complication being that
we allow systems on different algebras with different coordinate systems).

\begin{theorem}
\label{metries}The functions $W$ and $W_{\sigma}$ are asymmetric
pseudometrics, while $W_{\sigma\sigma}$ is a pseudometric. In addition,
optimal transport plans always exist for $W$, $W_{\sigma}$ and $W_{\sigma
\sigma}$.
\end{theorem}

\begin{proof}
We prove it for $W_{\sigma\sigma}$. The other cases are similar but simpler.
By its definition, $W_{\sigma\sigma}$ is real-valued and never negative, since
by Kadison's inequality $E_{\omega}(a^{\ast}a)\geq E_{\omega}(a)^{\ast
}E_{\omega}(a)$ we have
\begin{align}
&  \mu(a^{\ast}a)+\nu(b^{\ast}b)-\nu(E_{\omega}(a)^{\ast}b)-\nu(b^{\ast
}E_{\omega}(a))\label{Kadison}\\
&  =\nu\left(  \left|  b-E_{\omega}(a)\right|  ^{2}\right)  +\nu\left(
E_{\omega}(\left|  a\right|  ^{2})-\left|  E_{\omega}(a)\right|  ^{2}\right)
\nonumber\\
&  \geq\nu\left(  \left|  b-E_{\omega}(a)\right|  ^{2}\right) \nonumber\\
&  \geq0\nonumber
\end{align}
for all $a\in A$ and $b\in B$. (Alternatively, one can use Remark \ref{Ialt},
but (\ref{Kadison}) will come up again in the next section.) Note that for
$\omega=\delta_{\mu}$ we have $E_{\omega}=\operatorname{id}_{A}$, from which
$\omega\in T_{\sigma\sigma}(\mathbf{A},\mathbf{A})$ follows trivially by the
test in (\ref{BalE}). Consequently $W_{\sigma\sigma}(\mathbf{A},\mathbf{A}%
)=I_{\mathbf{A},\mathbf{A}}(\delta_{\mu})=0$ from (\ref{algKoste}).

\emph{The triangle inequality.} For $\omega\in T_{\sigma\sigma}(\mathbf{A}%
,\mathbf{B})$ and $\psi\in T_{\sigma\sigma}(\mathbf{B},\mathbf{C})$, we define
$\varphi=\omega\circ\psi$ through
\[
E_{\omega\circ\psi}=E_{\psi}\circ E_{\omega}.
\]
Then $\varphi\in T_{\sigma\sigma}(\mathbf{A},\mathbf{C})$, since $E_{\varphi
}\circ\alpha_{\upsilon,z}=\gamma_{\upsilon,z}\circ E_{\varphi}$, hence
$\mathbf{A}\varphi\mathbf{C}$, and similarly $\mathbf{A}^{\sigma}%
\varphi\mathbf{C}^{\sigma}$ and $(A,\sigma^{\mu},\mu)\varphi\left(
B,\sigma^{\xi},\xi\right)  $. As in the proof of \cite[Proposition 4.3]{D1},
but allowing multiple algebras as well as $k$, $l$ and $m$ (rather than just
$k$), and using Remark \ref{Ialt}, while skipping some details explained in
\cite{D1}, we have
\begin{align*}
I_{\mathbf{A},\mathbf{C}}(\varphi)^{1/2}  
&  =
\left\|  
\pi_{\mu}^{\varphi}(k)\Phi-\pi_{\xi}^{\varphi}(m)\Phi
\right\|  _{\oplus\varphi}\\
&  \leq
\left\|  
\pi_{\mu}^{\omega}(k)\Omega-\pi_{\nu}^{\omega}(l)\Omega
\right\|  _{\oplus\omega}
+
\left\|  
\pi_{\nu}^{\psi}(l)\Psi-\pi_{\xi}^{\psi}(m)\Psi
\right\|  _{\oplus\psi}\\
&  =
I_{\mathbf{A},\mathbf{B}}(\omega)^{1/2}
+
I_{\mathbf{B},\mathbf{C}}(\psi)^{1/2}
\end{align*}
by employing the triangle inequality in $\bigoplus_{i=1}^{d}(H_{\omega}%
\otimes_{\nu}H_{\psi})$ as well as properties of the relative tensor product
$H_{\omega}\otimes_{\nu}H_{\psi}$ of the $A$-$B$-bimodule $H_{\omega}$ and
$B$-$C$-bimodule $H_{\psi}$. The role of $H_{\omega}\otimes_{\nu}H_{\psi}$
(see \cite[Appendix V.B]{Con94}, \cite{Fal}, \cite{Sa} and \cite[Section
IX.3]{T2} for background) is to ensure that the ``middle term'' coming from
the difference between $\pi_{\nu}^{\omega}(l)\Omega$ and $\pi_{\nu}^{\psi
}(l)\Psi$ via imbedding into $\bigoplus_{i=1}^{d}(H_{\omega}\otimes_{\nu
}H_{\psi})$, indeed disappears. Now take the infimum on the left over all
$\varphi\in T_{\sigma\sigma}(\mathbf{A},\mathbf{C})$, which includes the
compositions $\omega\circ\psi$ for all $\omega\in T_{\sigma\sigma}%
(\mathbf{A},\mathbf{B})$ and $\psi\in T_{\sigma\sigma}(\mathbf{B},\mathbf{C}%
)$, followed in turn by the infima over all $\omega\in T_{\sigma\sigma
}(\mathbf{A},\mathbf{B})$ and $\psi\in T_{\sigma\sigma}(\mathbf{B}%
,\mathbf{C})$ on the right.

\emph{Symmetry.} This is essentially verbatim as in the proof of \cite[Theorem
3.9]{D2}, but in terms of $E_{\omega}:A\rightarrow B$, instead of involving
just one von Neumann algebra. However, note that the proof does not hold for
$W$ and $W_{\sigma}$, just for $W_{\sigma\sigma}$, as the KMS duals play a key
role here.

\emph{Optimal transport plans exist.}
Without loss (see for example \cite[Proposition 4.1]{D(B2)}) we can view each
element of $T_{\sigma\sigma}(\mathbf{A},\mathbf{B})$ as a state on the maximal
C*-tensor product $A\otimes_{\text{max}}B^{\prime}$. Consequently
$T_{\sigma\sigma}(\mathbf{A},\mathbf{B})$ is weakly* compact.
From $W_{\sigma\sigma}$'s definition there is a sequence $\omega_{q}\in
T_{\sigma\sigma}(\mathbf{A},\mathbf{B})$ such that $I_{\mathbf{A},\mathbf{B}%
}(\omega_{q})^{1/2}\rightarrow W_{\sigma\sigma}(\mathbf{A},\mathbf{B})$, and
necessarily having a weak* cluster point $\omega\in T_{\sigma\sigma
}(\mathbf{A},\mathbf{B})$. As a result, $\omega$ is an optimal transport plan
for $W_{\sigma\sigma}(\mathbf{A},\mathbf{B})$, by the following approximation:

Given $\varepsilon>0$, there is a $q_{0}$ such that
\[
\left|  I_{\mathbf{A},\mathbf{B}}(\omega_{q})-W_{\sigma\sigma}(\mathbf{A}%
,\mathbf{B})^{2}\right|  <\varepsilon
\]
for all $q>q_{0}$. Furthermore, there exist $b_{1}^{\prime},...,b_{d}^{\prime
}\in B^{\prime}$ such that%
\[
\left\|  l_{i}\Lambda_{\nu}-b_{i}^{\prime\ast}\Lambda_{\nu}\right\|
<\varepsilon
\]
for $i=1,...,d$. In addition there is a $q>q_{0}$ such that
\[
\left|  \omega_{q}(k_{i}\otimes b_{i}^{\prime})-\omega(k_{i}\otimes
b_{i}^{\prime})\right|  <\varepsilon
\]
for $i=1,...,d$. Then as in the proof of \cite[Lemma 6.2]{D1}, using
(\ref{algKoste}) and (\ref{omegaItv_delta&Eomega}), we find
\begin{align*}
\left|  
I_{\mathbf{A},\mathbf{B}}(\omega)
-
I_{\mathbf{A},\mathbf{B}}(\omega_{q})
\right|   
&  \leq
2\sum_{i=1}^{d}
\left|  
\nu(l_{i}^{\ast}E_{\omega_{q}}(k_{i}))
-
\nu(l_{i}^{\ast}E_{\omega}(k_{i}))
\right| \\
&  =
2\sum_{i=1}^{d}
\left|  
\left\langle 
l_{i}\Lambda_{\nu},
(E_{\omega_{q}}(k_{i}) - E_{\omega}(k_{i}))\Lambda_{\nu}
\right\rangle 
\right| \\
&  <
4\varepsilon\sum_{i=1}^{d}\left\|  k_{i} \right\|  + 2d\varepsilon.
\end{align*}
Thus
\[
\left|  
I_{\mathbf{A},\mathbf{B}}(\omega)
-
W_{\sigma\sigma}(\mathbf{A},\mathbf{B})^{2}
\right|  
<
4\varepsilon\sum_{i=1}^{d}\left\|  k_{i} \right\|
+
2d\varepsilon + \varepsilon
\]
for all $\varepsilon>0$.
\end{proof}

The remaining metric issue, namely under which conditions zero Wasserstein
distance between $\mathbf{A}$ and $\mathbf{B}$ implies that they are in fact
the same system (at least up to isomorphism), will be returned to in Corollary
\ref{MetrieseGetrouheid}.

We note that optimal transport plans in general need not be unique.

\begin{example}
In (\ref{algKoste}), take $A=B$, $\mu=\nu$ and $k_{i}=l_{i}$ for $i=1,...,d$.
Assume that the dynamics of $\mathbf{A}$ is trivial (i.e., there is none, or
each is taken as the identity map). Let $F$ be any von Neumann subalgebra of
$A$ containing $\{k_{1},...,k_{d}\}$ such that $\sigma_{t}^{\mu}(F)=F$.
Consider the unique conditional expectation $E_{F}:A\rightarrow F$ such that
$\mu\circ E_{F}=\mu$; see \cite{T72} and \cite[Theorem IX.4.2]{T2}. Then
$\omega_{F}$ defined from $E_{F}$ through (\ref{omegaItv_delta&Eomega})
trivially delivers $I_{\mathbf{A},\mathbf{A}}(\omega_{F})=0$ simply because
$E_{F}(k_{i})=k_{i}$. This is true for any such $F$, hence in general multiple
transport plans give the optimal cost $0$. A simple instance is where $\mu$ is
tracial, in particular when $A$ is abelian, in which case $\sigma^{\mu}$ is
trivial and the condition $\sigma_{t}^{\mu}(F)=F$ is automatically satisfied.
\end{example}

What remains is to study relevant properties of these Wasserstein distances
and take initial steps in showing how the general theory of these distances
can be used to study the structure and properties of systems. This is what we
turn to next.

\section{Properties and implications of Wasserstein
distances\label{AfdWEiensk}}

Here we derive properties of Wasserstein distances relevant to the structure
and characteristics of systems. These properties include symmetries related to
dualities of systems. The symmetries along with the metric properties, in
particular the triangle inequality, allow us to easily find simple bounds on
how far a (quantum) system is from detailed balance, in terms of its
Wasserstein distance from other systems satisfying detailed balance
(Subsection \ref{OndAfdAfwVanFB}). The idea is that these other systems will
typically be chosen as well understood or simpler, even classical, systems.
This opens the door to analyzing a system via simpler systems, in line with
the examples in Sections \ref{AfdKlasVb} and \ref{AfdKwantVb}. In addition
(Subsection \ref{OndAfdW&GemStruk}) we also study implications of Wasserstein
distances for the structure of systems, focussing in particular on the case of
zero Wasserstein distance. The goal is to show that Wasserstein distance can
give us information about common structure in two systems.

We continue with the notation from the previous section, again fixing
$\Upsilon$, the $Z_{\upsilon}$'s, and $d$ for all systems to be considered.

\subsection{Bounds on deviation from detailed balance\label{OndAfdAfwVanFB}}

The first step is to define two more duals of systems, in addition to the
KMS-dual from Definition \ref{stelselKMS-duaal}. The first of these extends
the dual defined in Subsection \ref{OndAfdDef}:

\begin{definition}
The \emph{dual} of $\mathbf{A}$ is defined to be $\mathbf{A}^{\prime
}:=(A^{\prime},\alpha^{\prime},\mu^{\prime},k^{\prime})$ where $k^{\prime
}=(k_{1}^{\prime},...,k_{d}^{\prime}):=(j_{A}(k_{1}^{\ast}),...,j_{A}%
(k_{d}^{\ast}))$.
\end{definition}

It is easily seen that $\mathbf{A}^{\prime\prime}=\mathbf{A}$. For systems
$\mathbf{A}$ and $\mathbf{B}$, and $\omega\in T(\mu,\nu)$, a useful fact is%
\begin{equation}
\mathbf{A}\omega\mathbf{B}\Leftrightarrow\mathbf{B}^{\prime}\mathbb{\omega
}^{\prime}\mathbf{A}^{\prime}\Leftrightarrow\mathbf{B}^{\sigma}\mathbb{\omega
}^{\sigma}\mathbf{A}^{\sigma}, \label{duaalEkw}
\end{equation}
which was shown in \cite[Corollary 4.6]{DS2}, where $\omega^{\prime}$ and
$\omega^{\sigma}$ are respectively defined via
\[
\omega'(b' \otimes a)
=
\delta_{\mu'}(E'_{\omega}(b') \otimes a)
=
\delta_{\nu}(E_{\omega}(a) \otimes b')
=
\omega(a\otimes b')
\]
for all $a\in A$, $b'\in B'$, and
\[
\omega^{\sigma}(b\otimes a')
=
\delta_{\mu}(E_{\omega}^{\sigma}(b)\otimes a'),
\]
for all $a^{\prime}\in A^{\prime}$, $b\in B$, in terms of the correspondence
in Definition \ref{KanVsToest}. I.e., $E_{\omega^{\prime}}=E_{\omega}^{\prime
}$ and $E_{\omega^{\sigma}}=E_{\omega}^{\sigma}$.

The second of these duals, extending the reverse from Section \ref{AfdFB&KMS},
is defined in terms of a reversing operation as appearing in Section
\ref{AfdFB&KMS}. In fact more generally we also need appropriate reversals for
u.p. maps between von Neumann algebras, which we formalize as follows.

\begin{definition}
Given reversing operations $\theta_{\mu}$ and $\theta_{\nu}$ for $(A,\mu)$ and
$(B,\nu)$ respectively, as well as a u.p. map $E:A\rightarrow B$ satisfying
$\nu\circ E=\mu$, we define its $(\theta_{\mu},\theta_{\nu})$\emph{-KMS dual}
$E^{\leftarrow}:B\rightarrow A$ by%
\[
E^{\leftarrow}:=\theta_{\mu}\circ E^{\sigma}\circ\theta_{\nu}.
\]
For any $\omega\in T(\mu,\nu)$ we can consequently define $\omega^{\leftarrow
}\in T(\nu,\mu)$ via
\[
\omega^{\leftarrow}(b\otimes a^{\prime})=\delta_{\mu}(E_{\omega}^{\leftarrow
}(b)\otimes a^{\prime})
\]
for all $a^{\prime}\in A^{\prime}$, $b\in B$, in terms of Definition
\ref{KanVsToest}, i.e., $E_{\omega^{\leftarrow}}=E_{\omega}^{\leftarrow}$ .
(One could write $E^{\theta_{\mu},\theta_{\nu}}$ instead of $E^{\leftarrow}$,
but $\theta_{\mu}$ and $\theta_{\nu}$ will always be clear from context.)
\end{definition}

Note that $E^{\leftarrow}=(\theta_{\nu}\circ E\circ\theta_{\mu})^{\sigma}$,
since $\theta_{\mu}^{\sigma}=\theta_{\mu}$ and $\theta_{\nu}^{\sigma}%
=\theta_{\nu}$ similar to Section \ref{AfdFB&KMS}, hence $E^{\leftarrow
\leftarrow}=E$.

\begin{definition}
We call a system $\mathbf{A}$ \emph{reversible} if there is a distinguished
$\alpha_{\upsilon,z}=\theta_{\mathbf{A}}$ for some $\upsilon\in\Upsilon$ with
$Z_{\upsilon}=\{z\}$ a 1-point set, where $\theta_{\mathbf{A}}$ is a reversing
operation for $(A,\mu)$, as defined in Section \ref{AfdFB&KMS}. The
\emph{reverse} of $\mathbf{A}$ is then defined as the system%
\[
\mathbf{A}^{\leftarrow}=(A,\alpha^{\leftarrow},\mu,k)
\]
where $\alpha_{\upsilon,z}^{\leftarrow}=(\alpha_{\upsilon,z})^{\leftarrow}$
for every $z\in Z_{\upsilon}$ and $\upsilon\in\Upsilon$. A reversible system
$\mathbf{A}$ is said to satisfy $\theta_{\mathbf{A}}$\emph{-sqdb} if
$\mathbf{A}^{\leftarrow}=\mathbf{A}$, extending the definition from Section
\ref{AfdFB&KMS}.
\end{definition}

As in Section \ref{AfdFB&KMS}, one has $\mathbf{A}^{\leftarrow\leftarrow
}=\mathbf{A}$. In this respect note that
\[
\theta_{\mathbf{A}}^{\leftarrow}=\theta_{\mathbf{A}},
\]
i.e., $\mathbf{A}^{\leftarrow}$ has the same reversing operation as
$\mathbf{A}$, since $\theta_{\mathbf{A}}^{\sigma}=\theta_{\mathbf{A}}$. We aim
to connect our Wasserstein distances to $\theta_{\mathbf{A}}$-sqdb.

We now have the following simple proposition in the vein of (\ref{duaalEkw}).

\begin{proposition}
\label{duaalOmkEkw}For reversible systems $\mathbf{A}$ and $\mathbf{B}$, one
has
\[
\mathbf{A}\omega\mathbf{B}\Leftrightarrow\mathbf{B}^{\leftarrow}%
\mathbb{\omega}^{\leftarrow}\mathbf{A}^{\leftarrow}.
\]
Note in addition that if $\omega\in T(\mathbf{A},\mathbf{B})$ for reversible
$\mathbf{A}$ and $\mathbf{B}$, then $\theta_{\mathbf{B}}\circ E_{\omega}%
\circ\theta_{\mathbf{A}}=E_{\omega}$, hence $E_{\omega}^{\leftarrow}%
=E_{\omega}^{\sigma}$ and $\omega^{\leftarrow}=\omega^{\sigma}$.
\end{proposition}

\begin{proof}
Note that $E_{\omega}\circ\alpha_{\upsilon,z}=\beta_{\upsilon,z}\circ
E_{\omega}$ is equivalent to $\alpha_{\upsilon,z}^{\sigma}\circ E_{\omega
}^{\sigma}=E_{\omega}^{\sigma}\circ\beta_{\upsilon,z}^{\sigma}$. Hence, for
reversible systems, $E_{\omega}\circ\alpha_{\upsilon,z}=\beta_{\upsilon
,z}\circ E_{\omega}$ is also equivalent to $\alpha_{\upsilon,z}^{\leftarrow
}\circ E_{\omega}^{\leftarrow}=E_{\omega}^{\leftarrow}\circ\beta_{\upsilon
,z}^{\leftarrow}$, since $\theta_{\mathbf{A}}\circ\theta_{\mathbf{A}%
}=\operatorname*{id}_{A}$ and $\theta_{\mathbf{B}}\circ\theta_{\mathbf{B}%
}=\operatorname*{id}_{B}$. The last statement in the proposition follows from
$T(\mathbf{A},\mathbf{B})$'s definition and the fact that the reversing
operations are part of the systems' dynamics.
\end{proof}

In what follows we implicitly assume that we choose our set of systems $X$
from Section \ref{AfdW} to contain all the duals that will be referred to.
Then the Wasserstein distances $W_{\sigma}$ and $W_{\sigma\sigma}$ have the
following symmetries with respect to the three types of duals discussed above.

\begin{theorem}
\label{Duaal&Sim}For any $\mathbf{A}$ and $\mathbf{B}$ we have
\[
W_{\sigma}(\mathbf{A},\mathbf{B})=W_{\sigma}(\mathbf{B}^{\prime}%
,\mathbf{A}^{\prime})=W_{\sigma}(\mathbf{B}^{\sigma},\mathbf{A}^{\sigma}).
\]
For reversible systems $\mathbf{A}$ and $\mathbf{B}$ we in addition have
\[
W_{\sigma}(\mathbf{A},\mathbf{B})=W_{\sigma}(\mathbf{B}^{\leftarrow
},\mathbf{A}^{\leftarrow}).
\]
The corresponding results also hold for $W_{\sigma\sigma}$, which can be written as 
$W_{\sigma\sigma}(\mathbf{A},\mathbf{B})
=
W_{\sigma\sigma}(\mathbf{A}',\mathbf{B}')$
etc., due to $W_{\sigma\sigma}$'s
usual pseudometric symmetry (Theorem \ref{metries}).
\end{theorem}

\begin{proof}
For $\omega$ $\in T_{\sigma}(\mathbf{A},\mathbf{B})$ (recall Definition
\ref{T(A,B)})%
\begin{align*}
I_{\mathbf{B}^{\prime},\mathbf{A}^{\prime}}(\omega^{\prime})  &  =\sum
_{i=1}^{d}\left[  \nu^{\prime}(l_{i}^{\prime}{}^{\ast}l_{i}^{\prime}%
)+\mu^{\prime}(k_{i}^{\prime}{}^{\ast}k_{i}^{\prime})-\mu^{\prime}%
(E_{\omega^{\prime}}(l_{i}^{\prime})^{\ast}k_{i}^{\prime})-\mu^{\prime}%
(k_{i}^{\prime}{}^{\ast}E_{\omega^{\prime}}(l_{i}^{\prime}))\right] \\
&  =\sum_{i=1}^{d}\left[  \nu(l_{i}^{\ast}l_{i})+\mu(k_{i}^{\ast}k_{i}%
)-\mu(k_{i}^{\ast}E_{\omega}^{\sigma}(l_{i}))-\mu(E_{\omega}^{\sigma}%
(l_{i})^{\ast}k_{i})\right] \\
&  =I_{\mathbf{A},\mathbf{B}}(\omega),
\end{align*}
using the definitions and properties of the various objects appearing, along
with $\mu(aE_{\omega}^{\sigma}(b))=\nu(E_{\omega}(a)b)$, which follows from
$E_{\omega}\circ\sigma_{t}^{\mu}=\sigma_{t}^{\nu}\circ E_{\omega}$ (see
\cite[Lemma 5.2]{D1} as well as the proof of \cite[Lemma 5.3]{D1}). Similarly
$I_{\mathbf{B}^{\sigma},\mathbf{A}^{\sigma}}(\omega^{\sigma})=I_{\mathbf{A}%
,\mathbf{B}}(\omega)$. Because of (\ref{duaalEkw}) applied to the modular
dynamics, along with $(\sigma_{t}^{\mu})^{\prime}=\sigma_{t}^{\mu^{\prime}}$
and $(\sigma_{t}^{\mu})^{\sigma}=\sigma_{-t}^{\mu}$, to verify the modular
transport plan properties, it follows that for $\omega\in T(\mathbf{A}%
,\mathbf{B})$ we have $\omega\in T_{\sigma}(\mathbf{A},\mathbf{B})$ if and
only if $\omega^{\prime}\in T_{\sigma}(\mathbf{B}^{\prime},\mathbf{A}^{\prime
})$, and if and only if $\omega^{\sigma}\in T_{\sigma}(\mathbf{B}^{\sigma
},\mathbf{A}^{\sigma})$. Consequently, from the definition of $W_{\sigma}$ we
have $W_{\sigma}(\mathbf{A},\mathbf{B})=W_{\sigma}(\mathbf{B}^{\prime
},\mathbf{A}^{\prime})=W_{\sigma}(\mathbf{B}^{\sigma},\mathbf{A}^{\sigma})$.

For reversible $\mathbf{A}$ and $\mathbf{B}$ and $\omega$ $\in T_{\sigma
}(\mathbf{A},\mathbf{B})$, one has $\mu(aE_{\omega}^{\leftarrow}%
(b))=\mu(aE_{\omega}^{\sigma}(b))=\nu(E_{\omega}(a)b)$ by Proposition
\ref{duaalOmkEkw}. Similar to above, it follows that $I_{\mathbf{B}%
^{\leftarrow},\mathbf{A}^{\leftarrow}}(\omega^{\leftarrow})=$ $I_{\mathbf{A}%
,\mathbf{B}}(\omega)$. Since $\mu\circ\theta_{\mathbf{A}}=\mu$, we have
$\sigma_{t}^{\mu}\circ\theta_{\mathbf{A}}=\theta_{\mathbf{A}}\circ\sigma
_{t}^{\mu}$ (similar to the proof of \cite[Proposition 6.4]{DS2} for example),
hence $(\sigma_{t}^{\mu})^{\leftarrow}=\sigma_{-t}^{\mu}$, which along with
Proposition \ref{duaalOmkEkw} tells us that for $\omega\in T(\mathbf{A}%
,\mathbf{B})$ we have $\omega\in T_{\sigma}(\mathbf{A},\mathbf{B})$ if and
only if $\omega^{\leftarrow}\in T_{\sigma}(\mathbf{B}^{\leftarrow}%
,\mathbf{A}^{\leftarrow})$. Thus $W_{\sigma}(\mathbf{B}^{\leftarrow
},\mathbf{A}^{\leftarrow})=W_{\sigma}(\mathbf{A},\mathbf{B})$.

The case of $\ W_{\sigma\sigma}$ is very similar. One has to verify that
$(\mathbf{A}^{\prime})^{\sigma}=(\mathbf{A}^{\sigma})^{\prime}$ and
$(\mathbf{A}^{\leftarrow})^{\sigma}=(\mathbf{A}^{\sigma})^{\leftarrow}$, which
are straightforward from the definitions and general properties of these
constructions, including $(j_{\nu}\circ E\circ j_{\mu})^{\prime}=j_{\mu}\circ
E^{\prime}\circ j_{\nu}$ (see \cite[Proposition 2.8]{DS2}). This allows us to
check that for $\omega\in T(\mathbf{A},\mathbf{B})$ and reversible
$\mathbf{A}$ and $\mathbf{B}$, one has $\omega\in T_{\sigma\sigma}%
(\mathbf{A},\mathbf{B})$ if and only if $\omega^{\leftarrow}\in T_{\sigma
\sigma}(\mathbf{B}^{\leftarrow},\mathbf{A}^{\leftarrow})$. Consequently,
$W_{\sigma\sigma}(\mathbf{B}^{\leftarrow},\mathbf{A}^{\leftarrow}%
)=W_{\sigma\sigma}(\mathbf{A},\mathbf{B})$, but by Theorem \ref{metries},
$W_{\sigma\sigma}(\mathbf{B}^{\leftarrow},\mathbf{A}^{\leftarrow}%
)=W_{\sigma\sigma}(\mathbf{A}^{\leftarrow},\mathbf{B}^{\leftarrow})$.
Similarly for the other two dualities.
\end{proof}

This result along with the basic metric properties in Theorem \ref{metries},
lead to simple consequences for quantifying and bounding the deviation of a
system from satisfying a detailed balance condition, in terms of the
Wasserstein distance of the system from another system which does satisfy some
detailed balance condition.

\begin{corollary}
\label{fbAfskat}Consider reversible systems $\mathbf{A}$ and $\mathbf{B}$. If
$\mathbf{B}$ satisfies $\theta_{\mathbf{B}}$-sqdb, then
\[
W_{\sigma}(\mathbf{A},\mathbf{A}^{\leftarrow})
\leq
2W_{\sigma}(\mathbf{A},\mathbf{B})
\text{ \ \ and \ \ }
W_{\sigma}(\mathbf{A}^{\leftarrow},\mathbf{A})
\leq
2W_{\sigma}(\mathbf{B},\mathbf{A})
\]
and
\[
W_{\sigma\sigma}(\mathbf{A},\mathbf{A}^{\leftarrow})
\leq
2W_{\sigma\sigma}(\mathbf{A},\mathbf{B}).
\]
\end{corollary}

\begin{proof}
By the triangle inequality (from Theorem \ref{metries})
\[
W_{\sigma}(\mathbf{A},\mathbf{A}^{\leftarrow})
\leq
W_{\sigma}(\mathbf{A},\mathbf{B})
+
W_{\sigma}(\mathbf{B,A}^{\leftarrow})
=
2W_{\sigma}(\mathbf{A},\mathbf{B}).
\]
Similarly for the other cases.
\end{proof}

Keep in mind that $\mathbf{A}^{\leftarrow}=\mathbf{A}$ would tell us that
$\mathbf{A}$ satisfies $\theta_{\mathbf{A}}$-sqdb. This condition indeed
follows from $W_{\sigma}(\mathbf{A},\mathbf{A}^{\leftarrow})=0$ if $k$
generates $A$ and $\{k_{1}^{\ast},...,k_{d}^{\ast}\}=\{k_{1},...,k_{d}\}$, as
can be seen from \cite[Theorems 3.9 and 3.10]{D2} and \cite[Section 6]{D1} as
well as the arguments used there (those papers do not explicitly cover
$W_{\sigma}$ as defined in this paper). Similarly for the other cases in the
corollary above. For $W_{\sigma}$ and $W_{\sigma\sigma}$ we in fact discuss
this from a more general point of view in Subsection \ref{OndAfdW&GemStruk}
(Theorem \ref{gemeneStruktuur} and Corollary \ref{MetrieseGetrouheid}).

Since we do not have corresponding symmetries for $W$ (from the proof of
Theorem \ref{Duaal&Sim} it is clear that the modular dynamics and KMS duals
played a key role in obtaining those symmetries, and $W$ simply does not have
the required structure built in) we only have the triangle inequality%
\[
W(\mathbf{A},\mathbf{A}^{\leftarrow})\leq W(\mathbf{A},\mathbf{B}%
)+W(\mathbf{B},\mathbf{A}^{\leftarrow})
\]
and similarly for $W(\mathbf{A}^{\leftarrow},\mathbf{A})$. However, the
central point remains the same, namely to be able to bound the deviation of
$\mathbf{A}$ from detailed balance in terms of its Wasserstein distance from
another system.

\begin{example}
Applying the corollary when $\mathbf{B}$ is classical, i.e., $B$ is abelian,
we take $\mathbf{B}$'s reversing operation to be the identity map, in which
case $\mathbf{B}^{\leftarrow}=\mathbf{B}^{\prime}$, and $\theta_{\mathbf{B}}%
$-sqdb simply says that $\mathbf{B}^{\prime}=\mathbf{B}$ (refer to Section
\ref{AfdFB&KMS}), which generalizes detailed balance of a Markov chain. Thus
we have the deviation $W_{\sigma}(\mathbf{A},\mathbf{A}^{\leftarrow})$ of
$\mathbf{A}$ from $\theta_{\mathbf{A}}$-sqdb bounded in terms of its distance
$W_{\sigma}(\mathbf{A},\mathbf{B})$ from a system $\mathbf{B}$ satisfying
usual classical detailed balance. Similarly for $W_{\sigma}(\mathbf{A}%
^{\leftarrow},\mathbf{A})$.
\end{example}

When both $\mathbf{A}$ and $\mathbf{B}$ are classical, the triviality of the
modular dynamics gives $W_{\sigma}=W$. Hence the inequalities for $W_{\sigma}$
can be stated as%
\[
W(\mathbf{A},\mathbf{A}^{\prime})\leq2W(\mathbf{A},\mathbf{B})\text{ \ \ and
\ \ }W(\mathbf{A}^{\prime},\mathbf{A})\leq2W(\mathbf{B},\mathbf{A}).
\]
As illustrated in Subsection \ref{OndAfdKlasVbUitdrGeval}, $W$ is not
symmetric in general, even for classical systems. To ensure the usual metric
symmetry, $W_{\sigma\sigma}$ still has to be used.

\textbf{Conclusion.}
The various inequalities above illustrate that considering a set of systems at
most some specified Wasserstein distance away from a given system $\mathbf{B}$
satisfying a detailed balance property, sensibly relates to the deviation (as
measured by Wasserstein distance) of systems in that set from detailed
balance. In addition, we can choose $\mathbf{B}$ to be simple or well
understood, say some appropriately chosen classical system. In this way we can
obtain what we expect to be sets of non-equilibrium systems $\mathbf{A}$
having structure close enough to detailed balance to make them amenable to
further analysis.

\subsection{Zero Wasserstein distance and common
structure\label{OndAfdW&GemStruk}}

To shed light on this issue regarding the structure of a system $\mathbf{A}$,
we turn to the implications of zero Wasserstein distance in relation to common
structure in two systems. Our basic result in this regard is stated as Theorem
\ref{gemeneStruktuur} below. This in fact boils down to a generalization of
the ``faithfulness'' of an (asymmetric) metric $\rho$, namely the property
that $\rho(x,y)=0$ implies $x=y$, which was not discussed in the previous
section. We leave the case of non-zero Wasserstein distance, which appears to
much more involved, for future work. Again modular properties and KMS duals
play a key role, hence our strongest results will be for $W_{\sigma}$ and
$W_{\sigma\sigma}$.

This is intimately related to the coordinates of a system, in particular the
algebra generated by them, hence we introduce the following terminology and
notation to be used in the remainder of this section.

\begin{definition}
The \emph{coordinate algebra} $M$ of a system $\mathbf{A}$ is the von Neumann
subalgebra of $A$ generated by $\{k_{1},...,k_{d}\}$.
\end{definition}

This notation $M$ will be used as a convention below. Similarly, $N$ will
denote the coordinate algebra of $\mathbf{B}$.

Given any $\mathbf{A}$, the condition $\{k_{1}^{\ast},...,k_{d}^{\ast
}\}=\{k_{1},...,k_{d}\}$ will now become relevant. (Note that this does not
mean that $k_{1}^{\ast}=k_{1}$, etc.) In \cite{D1, D2} this condition was
needed as part of the proof of faithfulness of an asymmetric metric. This will
now be obtained in a generalized form for $W_{\sigma}$ and $W_{\sigma\sigma}$.
For later convenience, we give this condition a name:

\begin{definition}
A system $\mathbf{A}$ is called \emph{hermitian} if $\{k_{1}^{\ast}%
,...,k_{d}^{\ast}\}=\{k_{1},...,k_{d}\}$.
\end{definition}

We need the following lemma extending a basic technical point from \cite[Lemma
6.3]{D1}.

\begin{lemma}
\label{nulKoste}Consider systems $\mathbf{A}$ and $\mathbf{B}$, with
$\mathbf{A}$ hermitian. Assume that there is an $\omega\in T(\mathbf{A}%
,\mathbf{B})$ such that $I_{\mathbf{A},\mathbf{B}}(\omega)=0$. Then the
restriction of $E_{\omega}$ to $M$ is a normal unital $\ast$-homomorphism
\[
E_{\omega}|_{M}:M\rightarrow N
\]
and $E_{\omega}(k_{i})=l_{i}$ for $i=1,...,d$.
\end{lemma}

\begin{proof}
As in (\ref{Kadison}),
\[
\nu\left(  \left|  l_{i}-E_{\omega}(k_{i})\right|  ^{2}\right)  +\nu\left(
E_{\omega}(\left|  k_{i}\right|  ^{2})-\left|  E_{\omega}(k_{i})\right|
^{2}\right)  =0,
\]
due to $I_{\mathbf{A},\mathbf{B}}(\omega)=0$. Thus $E_{\omega}(k_{i})=l_{i}$
and $E_{\omega}(\left|  k_{i}\right|  ^{2})=\left|  E_{\omega}(k_{i})\right|
^{2}$, since $\nu$ is faithful and $E_{\omega}(\left|  k_{i}\right|  ^{2}%
)\geq\left|  E_{\omega}(k_{i})\right|  ^{2}$. Let $M_{0}$ be the unital $\ast
$-algebra generated by $\{1_{A},k_{1},...,k_{d}\}$. Applying \cite[Theorem
3.1]{Choi74} and using the fact that $\mathbf{A}$ is hermitian, we find that
$E_{\omega}|_{M_{0}}:M_{0}\rightarrow N$ is a unital $\ast$-homomorphism. As
$E_{\omega}$ is normal (i.e., $\sigma$-weakly continuous) and $M_{0}$ is
$\sigma$-weakly dense in $M$, the lemma follows, simply because the maps
$a_{1}\mapsto a_{1}a_{2}$ and $a_{2}\mapsto a_{1}a_{2}$ are $\sigma$-weakly continuous.
\end{proof}

Note that this lemma is applicable when $W(\mathbf{A},\mathbf{B})=0$, due to
the existence of optimal transport plan in Theorem \ref{metries}, which indeed
led to $\mathbf{A}=\mathbf{B}$ for the case of $A=B$ and $k=l$ in
\cite[Theorem 3.10]{D2}. However, to make any appreciable further progress in
our current more general context, we need the additional structure of modular
and KMS transport plans in the definitions of $W_{\sigma}$ and $W_{\sigma
\sigma}$.

To formulate our results below more succinctly, we are going to use the
notation
\[
\alpha|_{M}%
\]
to refer to the restrictions $(\alpha|_{M})_{\upsilon,z}:=\alpha_{\upsilon
,z}|_{M}$ and
\[
E_{\omega}\circ\alpha
\]
to refer to $(E_{\omega}\circ\alpha)_{\upsilon,z}:=E_{\omega}\circ
\alpha_{\upsilon,z}$, etc., for all $z\in Z_{\upsilon}$ and $\upsilon
\in\Upsilon$, in line with Definition \ref{stelsel}. I.e., we are going to
suppress $\upsilon$ and $z$ in our notation.

Also recall that a $\ast$-isomorphism from one von Neumann algebra to another
is a linear map that is one to one (injective), onto (surjective) and
preserves all algebraic structure.

\begin{theorem}
\label{gemeneStruktuur}Let $\mathbf{A}$ and $\mathbf{B}$ be hermitian systems
such that $W_{\sigma}(\mathbf{A},\mathbf{B})=0$. Then there is a $\ast
$-isomorphism $\iota_{\mathbf{A},\mathbf{B}}:M\rightarrow N$ between the
coordinate algebras of $\mathbf{A}$ and $\mathbf{B}$, uniquely determined by
$\iota_{\mathbf{A},\mathbf{B}}(k_{i})=l_{i}$ for $i=1,...,d$, such that
\begin{equation}
E_{\omega}\circ\alpha|_{M}=\beta\circ\iota_{\mathbf{A},\mathbf{B}}
\label{betaBeperk}%
\end{equation}
for any optimal transport plan $\omega\in T_{\sigma}(\mathbf{A},\mathbf{B})$.
\end{theorem}

\begin{proof}
By Theorem \ref{metries} and Lemma \ref{nulKoste} we have a normal unital
$\ast$-homomorphism $\iota_{\mathbf{A},\mathbf{B}}:=E_{\omega}|_{M}%
:M\rightarrow N$, necessarily uniquely determined by $\iota_{\mathbf{A}%
,\mathbf{B}}(k_{i})=l_{i}$ for $i=1,...,d$. As in Theorem \ref{Duaal&Sim}'s
proof we have $I_{\mathbf{B}^{\sigma}\mathbf{,A}^{\sigma}}(\omega^{\sigma}%
)=0$. We therefore also have a normal unital $\ast$-homomorphism
$\iota_{\mathbf{B}^{\sigma}\mathbf{,A}^{\sigma}}:=E_{\omega}^{\sigma}%
|_{M}:N\rightarrow M$, uniquely determined by $\iota_{\mathbf{B}^{\sigma
}\mathbf{,A}^{\sigma}}(l_{i})=k_{i}$ for $i=1,...,d$. From this it is easily
seen that $\iota_{\mathbf{A},\mathbf{B}}$ and $\iota_{\mathbf{B}^{\sigma
}\mathbf{,A}^{\sigma}}$ are each other's inverses: $\iota_{\mathbf{B}^{\sigma
}\mathbf{,A}^{\sigma}}\circ\iota_{\mathbf{A},\mathbf{B}}(k_{i})=k_{i}$ meaning
that $\iota_{\mathbf{B}^{\sigma}\mathbf{,A}^{\sigma}}\circ\iota_{\mathbf{A}%
,\mathbf{B}}|_{M_{0}}=\operatorname*{id}_{M_{0}}$ with $M_{0}$ as in the
lemma's proof, hence $\iota_{\mathbf{B}^{\sigma}\mathbf{,A}^{\sigma}}%
\circ\iota_{\mathbf{A},\mathbf{B}}=\operatorname*{id}_{M}$ by $\sigma$-weak
continuity; similarly $\iota_{\mathbf{A},\mathbf{B}}\circ\iota_{\mathbf{B}%
^{\sigma}\mathbf{,A}^{\sigma}}=\operatorname*{id}_{N}$. Hence $\iota
_{\mathbf{A},\mathbf{B}}$ is indeed a $\ast$-isomorphism from $M$ onto $N$.

Now $E_{\omega}\circ\alpha|_{M}=\beta\circ\iota_{\mathbf{A},\mathbf{B}}$
follows directly from $E_{\omega}\circ\alpha=\beta\circ E_{\omega}$, which is
part of $T_{\sigma}(\mathbf{A},\mathbf{B})$'s definition; see (\ref{BalE}) and
Definition \ref{T(A,B)}.
\end{proof}

In particular (\ref{betaBeperk}) tells us (by applying $\iota_{\mathbf{B}%
^{\sigma}\mathbf{,A}^{\sigma}}$ from the right) that
\begin{equation}
\beta|_{N}=E_{\omega}\circ\alpha|_{M}\circ\iota_{\mathbf{B}^{\sigma
}\mathbf{,A}^{\sigma}}, \label{betaItvAlfa}%
\end{equation}
showing that $\beta|_{N}$ is determined by $\alpha|_{M}$, while $\alpha|_{M}$
is at least constrained by $\beta|_{N}$. This clearly illustrates our main
point about common structure in the dynamics of the two systems when
$W_{\sigma}(\mathbf{A},\mathbf{B})=0$.

Building on this result, and in preparation for the corollaries below, we
state the following:

\begin{definition}
Two systems $\mathbf{A}$ and $\mathbf{B}$ are called \emph{isomorphic} if
there is a $\ast$-isomorphism $\iota:A\rightarrow B$ such that $\iota
\circ\alpha=\beta\circ\iota$, $\nu\circ\iota=\mu$ and $\iota(k_{i})=l_{i}$ for
$i=1,...,d$.
\end{definition}

For all intents and purposes isomorphic systems are the same, but possibly
represented differently. Below we show how isomorphic systems can be
identified inside $\mathbf{A}$ and $\mathbf{B}$ in Theorem
\ref{gemeneStruktuur} when additional assumptions are made. In particular, in
Theorem \ref{gemeneStruktuur} we did not assume that $\alpha(M)\subset M$,
i.e., that $\alpha(M)$ is contained in $M$, or that $\beta(N)\subset N$, hence
we could not necessarily restrict the systems $\mathbf{A}$ or $\mathbf{B}$ to
systems on $M$ and $N$ respectively.

\begin{corollary}
\label{gemeneSubstelsel}Assuming that $\alpha(M)\subset M$, in addition to
Theorem \ref{gemeneStruktuur}'s assumptions, it follows that $\beta(N)\subset
N$ and
\[
\iota_{\mathbf{A},\mathbf{B}}\circ\alpha|_{M}=\beta\circ\iota_{\mathbf{A}%
,\mathbf{B}}.
\]
Consequently we can restrict all the structures in $\mathbf{A}$ and
$\mathbf{B}$ to $M$ and $N$ respectively to obtain isomorphic systems
$\mathbf{M}$ and $\mathbf{N}$ respectively.
\end{corollary}

\begin{proof}
This follows directly from Theorem \ref{gemeneStruktuur} and $\iota
_{\mathbf{A},\mathbf{B}}:=E_{\omega}|_{M}$ in its proof. Keep in mind that
$\nu\circ\iota_{\mathbf{A},\mathbf{B}}=\nu\circ E_{\omega}|_{M}=\mu|_{M}$; see
Definition \ref{KanVsToest}.
\end{proof}

This is a very clear cut case of common structure in $\mathbf{A}$ and
$\mathbf{B}$. It suggests that when we aim to analyze a system $\mathbf{B}$ by
comparison to simple or well understood systems, one strategy would be to
choose each of the latter as a system $\mathbf{A}$ with $A$ generated by
$\mathbf{A}$'s coordinates $k_{1},...,k_{d}$. For $W_{\sigma}(\mathbf{A}%
,\mathbf{B})=0$, this corollary then implies that
\[
\beta|_{N}=\iota_{\mathbf{A},\mathbf{B}}\circ\alpha\circ\iota_{\mathbf{A}%
,\mathbf{B}}^{-1},
\]
giving a precise expression of how $\mathbf{A}$ is contained in $\mathbf{B}$,
and how $\mathbf{A}$'s properties and behaviour are consequently reflected in
that of $\mathbf{B}$. Theorem \ref{gemeneStruktuur} can be viewed as weaker
version of this situation.

To show explicitly how Theorem \ref{gemeneStruktuur} relates to the
faithfulness of an asymmetric metric, we state the following special case of
this corollary:

\begin{corollary}
\label{MetrieseGetrouheid}Consider two hermitian systems $\mathbf{A}$ and
$\mathbf{B}$, where $A$ is generated by $k_{1},...,k_{d}$ and $B$ by
$l_{1},...,l_{d}$. If $W_{\sigma}(\mathbf{A},\mathbf{B})=0$, it then follows
that $\mathbf{A}$ and $\mathbf{B}$ are isomorphic.
\end{corollary}

Note that the relationship between $\alpha|_{M}$ and $\beta|_{N}$ in Theorem
\ref{gemeneStruktuur} and Corollary \ref{gemeneSubstelsel} is not symmetric,
as may be expected from the lack of symmetry of the asymmetric pseudometric
$W_{\sigma}$.

\begin{example}
In the example treated in Subsection \ref{OndAfdKlasVbUitdrGeval}, where
$N=B$, one can show that $W_{\sigma}(\mathbf{A},\mathbf{B})=0$ is not
sufficient to ensure that $\mathbf{B}$ is isomorphic to a system on $M$
obtained by restriction of $\mathbf{A}$. Indeed, $W_{\sigma}(\mathbf{A}%
,\mathbf{B})=0$ does not imply that $\alpha(M)\subset M$, despite
$\beta(N)\subset N$ being true.
This is closely related to (\ref{asimVbDeel1}) and (\ref{asimVbDeel2}); keep
in mind that $W=W_{\sigma}$ for classical systems, as the modular dynamics for
such systems are trivial.
\end{example}

To attain a symmetric relationship, we of course need to resort to the
(symmetric) pseudometric $W_{\sigma\sigma}$. One can indeed expect that the
requirement $W_{\sigma\sigma}(\mathbf{A},\mathbf{B})=0$ may place stronger
restrictions than $W_{\sigma}(\mathbf{A},\mathbf{B})=0$ on \thinspace$\alpha$,
given $\beta$, simply because by definition we always have $W_{\sigma
}(\mathbf{A},\mathbf{B})\leq W_{\sigma\sigma}(\mathbf{A},\mathbf{B})$.

\begin{corollary}
\label{simStruk}Let $\mathbf{A}$ and $\mathbf{B}$ be hermitian systems such
that $W_{\sigma\sigma}(\mathbf{A},\mathbf{B})=0$. Then
\[
E_{\omega}\circ\alpha|_{M}=\beta\circ\iota_{\mathbf{A},\mathbf{B}}\text{
\ \ and \ \ }E_{\omega}^{\sigma}\circ\beta|_{N}=\alpha\circ\iota
_{\mathbf{B},\mathbf{A}}%
\]
for any optimal transport plan $\omega\in T_{\sigma\sigma}(\mathbf{A}%
,\mathbf{B})$, in which case $\omega^{\sigma}\in T_{\sigma\sigma}%
(\mathbf{B},\mathbf{A})$ is also optimal.
\end{corollary}

\begin{proof}
Here $W_{\sigma}(\mathbf{A},\mathbf{B})=0$ and $W_{\sigma}(\mathbf{B}%
,\mathbf{A})=0$. The former follows from the definitions of $W_{\sigma}$ and
$W_{\sigma\sigma}$, which ensure that $W_{\sigma}(\mathbf{A},\mathbf{B})\leq
W_{\sigma\sigma}(\mathbf{A},\mathbf{B})$. The latter follows from the symmetry
of $W_{\sigma\sigma}$ given by Theorem \ref{metries}. By $T_{\sigma\sigma
}(\mathbf{A},\mathbf{B})$'s definition, we have $\omega^{\sigma}\in
T_{\sigma\sigma}(\mathbf{B},\mathbf{A})$ for any $\omega\in T_{\sigma\sigma
}(\mathbf{A},\mathbf{B})$. Hence for an optimal $\omega\in T_{\sigma\sigma
}(\mathbf{A},\mathbf{B})$, the cost $I_{\mathbf{B,A}}(\omega^{\sigma})$ is
defined and as in Theorem \ref{Duaal&Sim}'s proof we have $I_{\mathbf{B,A}%
}(\omega^{\sigma})=0$, so $\omega^{\sigma}$ is indeed optimal. Now we can
simply apply Theorem \ref{gemeneStruktuur} to both directions.
\end{proof}

We then have a corresponding symmetric version of Corollary
\ref{gemeneSubstelsel} as well, where we could assume either $\alpha(M)\subset
M$ or $\beta(N)\subset N$.

More generally, for any optimal transport plan, including for the case of $W$,
the condition $E_{\omega}\circ\alpha=\beta\circ E_{\omega}$ implies a relation
between $\mathbf{A}$ and $\mathbf{B}$ which can be viewed as a weaker
condition than an isomorphism between systems. This is the case even when the
systems are not hermitian or the Wasserstein distance being used is not zero.
Theorem \ref{gemeneStruktuur} and its corollaries simply state strong forms of
such a relation.

The general case can possibly be fruitfully viewed from the perspective of
normal u.c.p. maps as morphisms, or, in keeping with the bimodule approach,
from that of Connes' correspondences as morphisms. This will not be pursued
here, but see \cite[Appendix V.B]{Con94} and \cite[Section 5]{DS2} for
background on this matter, as well as \cite{BCM} in relation to dynamical
systems with automorphic (i.e., unitary) dynamics.

What is needed in addition to our development above, is a fitting notion of
the ``size'' of $E_{\omega}(A)$ in $B$, or of $E_{\psi}(B)$ in $A$, in
particular for optimal transport plans (for any of our Wasserstein distances)
$\omega$ from $\mathbf{A}$ to $\mathbf{B}$ or $\psi$ from $\mathbf{B}$ to
$\mathbf{A}$. This is relevant, simply because
\[
\beta(E_{\omega}(a))=E_{\omega}(\alpha(a))
\]
for all $a\in A$, hence $\alpha$ and the transport plan $\omega$ determines
the behaviour of $\beta$ on $E_{\omega}(A)$. The role of different optimal
transport plans, when not unique, may also be of interest.

\section{Outlook\label{AfdEinde}}

This paper developed the conceptual and technical foundations of a novel
framework for expressing and studying deviation from detailed balance via
optimal transport, in both the quantum and classical cases. In the process it
tied together the ideas built up in the papers \cite{DS2, D1, D2}, leading to
an optimal transport perspective on \cite{DS2}. Corollary \ref{fbAfskat}
together with its discussion is a fairly representative illustration of our
aims. More broadly the goal is to study (quantum) dynamical systems via the
optimal transport ideas and tools developed in this paper.

Much remains to be explored. In particular, to systematically apply the theory
to analyze systems, including in the context of non-equilibrium statistical
mechanics, as mentioned in the introduction. We expect that it should be of
broader use, though, for example in noncommutative as well as classical
ergodic theory to quantify and study how far a systems is from satisfying a
particular ergodic property.

The theory itself could also be developed further, some lines of investigation
having been discussed at the end of Section \ref{AfdWEiensk}.\textbf{ }In
addition, it would be worth investigating the case where the coordinates $k$
of a system $\mathbf{A}$ are not necessarily bounded. In this case they will
not be elements of $A$, but a natural assumption would be that they are
affiliated to $A$.

\bigskip
\textbf{Acknowledgments.} The support of the DSI-NRF Centre of Excellence 
in Mathematical and Statistical Sciences (CoE-MaSS) towards this research 
is hereby acknowledged. Opinions expressed and conclusions arrived at, are 
those of the authors and are not necessarily to be attributed to the CoE.

\end{document}